\documentclass[11pt]{article}


\usepackage{amsmath,amssymb,amsthm,mathtools}

    \usepackage{thmtools}
    \usepackage{thm-restate}
	\usepackage{a4,geometry}

\usepackage{graphicx}
\usepackage{paralist}
\usepackage{bm}
\usepackage{xspace}
\usepackage{url}
\usepackage{fullpage, prettyref}
\usepackage{boxedminipage}
\usepackage{wrapfig}
\usepackage{ifthen}
\usepackage{color}
\usepackage[usenames,dvipsnames]{xcolor}
\usepackage[colorlinks,citecolor=blue,linkcolor=BrickRed]{hyperref}
\usepackage{framed}

\usepackage{algpseudocode}
\usepackage[ruled,vlined, linesnumbered]{algorithm2e}
\usepackage{titlesec}
\titlespacing*{\subsubsection}{0pt}{0.2em}{0.2em}

\usepackage{thmtools}
\usepackage{thm-restate}
\usepackage{cleveref}

\newtheorem{theorem}{Theorem}[section]

\newtheorem{lemma}[theorem]{Lemma}

\newtheorem{corollary}[theorem]{Corollary}
\newtheorem{definition}[theorem]{Definition}
\newtheorem{proposition}[theorem]{Proposition}

\newcommand{\ignore}[1]{}

\newcommand{\poly}{\mathsf{poly}}

\newcommand{\bs}{\mathbf{s}}

\newcommand{\bx}{\boldsymbol{x}}

\newcommand{\bz}{\mathbf{z}}

\newcommand{\E}{\mathbb{E}}

\newcommand\restr[2]{{
  \left.\kern-\nulldelimiterspace
  #1 
  \vphantom{\big|} 
  \right|_{#2} 
 }}

\newcommand{\Sec}[1]{\hyperref[sec:#1]{\S\ref*{sec:#1}}} 
\newcommand{\Eqn}[1]{\hyperref[eq:#1]{(\ref*{eq:#1})}} 
\newcommand{\Fig}[1]{\hyperref[fig:#1]{Fig.\,\ref*{fig:#1}}}
\newcommand{\Tab}[1]{\hyperref[tab:#1]{Tab.\,\ref*{tab:#1}}}
\newcommand{\Thm}[1]{\hyperref[thm:#1]{Theorem\,\ref*{thm:#1}}} 
\newcommand{\Fact}[1]{\hyperref[fact:#1]{Fact\,\ref*{fact:#1}}} 
\newcommand{\Lem}[1]{\hyperref[lem:#1]{Lemma~\ref*{lem:#1}}} 
\newcommand{\Prop}[1]{\hyperref[prop:#1]{Proposition~\ref*{prop:#1}}} 
\newcommand{\Cor}[1]{\hyperref[cor:#1]{Corollary~\ref*{cor:#1}}} 
\newcommand{\Conj}[1]{\hyperref[conj:#1]{Conjecture~\ref*{conj:#1}}} 
\newcommand{\Rem}[1]{\hyperref[rem:#1]{Remark~\ref*{rem:#1}}} 
\newcommand{\Def}[1]{\hyperref[def:#1]{Definition~\ref*{def:#1}}} 
\newcommand{\Alg}[1]{\hyperref[alg:#1]{Alg.~\ref*{alg:#1}}} 
\newcommand{\Ex}[1]{\hyperref[ex:#1]{Ex.~\ref*{ex:#1}}} 
\newcommand{\Clm}[1]{\hyperref[clm:#1]{Claim~\ref*{clm:#1}}} 
\newcommand{\Step}[1]{\hyperref[step:#1]{Step~\ref*{step:#1}}} 
\newcommand{\Obs}[1]{\hyperref[obs:#1]{Observation~\ref*{obs:#1}}} 
\renewcommand{\Alg}[1]{\hyperref[alg:#1]{Algorithm~\ref*{alg:#1}}} 

\newcommand{\err}{\mathsf{err}}

\newcommand{\vectorBal}{\mathsf{TripletWalk}}

\graphicspath{{./Figures/}}

\renewcommand{\bar}{\overline}

\title{A Simpler Analysis of the Bansal–Jiang Quasi-Monte Carlo Algorithm via Haar Wavelets}

\date{}

\author{Jiaheng Chen\thanks{University of Chicago, Chicago, IL, USA. \texttt{jiaheng@uchicago.edu}.}
\and 
Agastya Vibhuti Jha\thanks{University of Chicago, Chicago, IL, USA. \texttt{agastyajha@uchicago.edu}.}
\and
Haotian Jiang\thanks{University of Chicago, Chicago, IL, USA. \texttt{jhtdavid@uchicago.edu}.}}

\begin{document}

\allowdisplaybreaks
\begin{titlepage}
\maketitle

\begin{abstract}
Numerical integration---approximating the integral of a function $f$ using $n$ point evaluations---is a central task in science and engineering. The two main paradigms for this problem, the Monte Carlo (MC) and quasi-Monte Carlo (QMC) methods, have distinct strengths and limitations, and a fundamental question is to design a method that combines the benefits of both.

\smallskip

Building on recent algorithmic advances in discrepancy theory, Bansal and Jiang \cite{BJ25a} gave a randomized QMC method that naturally bridges the MC and QMC error guarantees. Their method also achieves a surprising improvement over the classical Koksma--Hlawka inequality for QMC methods: it attains an error bound of $\widetilde{O}(\sigma_{\mathsf{SO}}(f)/n)$, where $\sigma_{\mathsf{SO}}(f)$ is a new notion of \emph{smoothed-out variation} that they introduced and showed to be substantially smaller than the Hardy--Krause variation governing the classical bound.

\smallskip

However, the analysis in \cite{BJ25a} is quite involved: it must carefully exploit the structure of the dyadic decomposition and the randomness of the algorithm inside a sufficiently fine discretization of the Hlawka--Zaremba formula to obtain cancellations among the high-frequency components in the Fourier decomposition of $f$. The contribution of this article is twofold:
\begin{enumerate}
    \item We give an equivalent characterization of  $\sigma_{\mathsf{SO}}(f)$ in terms of the Haar--Besov seminorm of $f$, relating this new notion of smoothed-out variation to classical quantities. 
    \item Through this characterization, we provide a conceptually simpler and more direct analysis of the Bansal--Jiang QMC method via Haar decomposition, bypassing the use of the  Hlawka--Zaremba formula, Fourier decomposition, and the delicate cancellation arguments of \cite{BJ25a} that heavily exploit the structure of dyadic decomposition.
\end{enumerate}
\end{abstract}

 \thispagestyle{empty}
\end{titlepage}

\thispagestyle{empty}
{\hypersetup{linkcolor=BrickRed}
 \tableofcontents
}

\thispagestyle{empty}
\newpage
\setcounter{page}{1}

\section{Introduction}\label{sec:intro}
Numerical integration is the problem of estimating an integral\footnote{Assuming the integration domain is $[0,1]^d$ is without loss of generality by standard reductions   \cite{DP10,Owe13}.}  $\smash{\bar{f} = \int_{\mathbf{z} \in [0,1]^{d}} f(\mathbf{z})\,d\mathbf{z}}$ by evaluating the given function $f$ on a finite set of points.

This problem arises in many applications \cite{WikiNumIntegration}, where $\overline{f}$ cannot be computed exactly because $f$ is too complicated, or may only be accessed through point evaluations. 
The goal of numerical integration is to choose a set of points $A \subset [0,1]^d$ and approximate $\overline{f}$ by the average $\smash{\bar{f}(A) = |A|^{-1}\sum_{\mathbf{z} \in A}f(\mathbf{z})}$, so that the error $\err(A,f) := \bar{f}(A) - \overline{f}$ is as small as possible. 
Broadly speaking, there are two paradigms for choosing the set $A$:

\smallskip
\noindent \textbf{Monte Carlo Methods.} Monte Carlo (MC) methods choose $A$ by sampling $n$ points independently and uniformly at random. 
It is known that the mean squared error here is 
\begin{equation}\label{eqn:MCError}
    \mathbb{E}[\mathsf{err}(A,f)^{2}] = \sigma^{2}(f) / n,
\end{equation}
where $\smash{\sigma(f) := \big(\E_{\mathbf{z}\sim[0,1]^d}(f(\mathbf{z}) - \overline{f})^2\big)^{1/2}}$  is the standard deviation of $f$.

As they require only random samples, Monte Carlo methods are flexible and useful in a wide range of applications, and they were named as one of the ten most influential algorithms of the twentieth century \cite{DS2000}. 

The main drawback of MC is its slow $n^{-1/2}$ convergence rate---to attain $\varepsilon$ error, one needs roughly $\sigma^{2}(f)/\varepsilon^{2}$ samples---in many applications, obtaining samples is expensive, and a larger sample size also increases the running time of algorithms that process it (e.g., see \cite{Glasserman03, Fishman05}).

\smallskip
\noindent \textbf{Quasi-Monte Carlo Methods.} 

Quasi-Monte Carlo (QMC) methods 
abandon the flexibility of using only random samples and instead, they carefully choose a deterministic set $A$.  
These methods can achieve a convergence rate of $\widetilde{O}_d(1/n)$\footnote{Throughout, $\widetilde{O}(\cdot)$ hides $\poly(\log n)$ factors, $O_d(\cdot)$ hides $\exp(O(d))$ factors,  and $\widetilde{O}_d(\cdot)$ hides $O_d(\log^{O(d)} n)$ factors.} in moderate dimensions---a quadratic speedup over MC that leads to their widespread applications \cite{Lem09,Mat09,DKPS13,Owe13}.

The key benchmark for QMC methods is several classical Koksma--Hlawka inequalities that bound $|\err(A,f)|$ by the product of a property of $f$ and a uniformity measure of $A$ \cite{Kok42,Hla61,Zar68}. For instance, the $\ell_2$ Koksma--Hlawka inequality states that  
\begin{equation}\label{eqn:Koksma-Hlawka}
|\err(A,f)| \leq \frac{1}{n} \cdot \sigma_{\mathsf{HK}}(f) \cdot D_2^*(A) ,
\end{equation}
where $\sigma_{\mathsf{HK}}(f)$ is the $\ell_2$ Hardy--Krause variation which measures the volatility of $f$, and $D_2^*(A)$ is the $\ell_2$ {\em star} discrepancy that measures the uniformity of the set $A$. \footnote{As an illustration in 1-d, $\sigma_{\mathsf{HK}}(f) =(\int_0^1 f'(z)^2\,dz)^{1/2}$ (under mild smoothness assumptions) is larger for faster-varying $f$; and $D_2^*(A) = (\int_0^1 D(z)^2\,d z)^{1/2}$, where $D(z) = |A| z - |A\cap [0,z)|$ measures the ``continuous discrepancy'' between the size of the
interval $[0,z)$ and the number of points from $A$ that lie in it. } 

The Koksma--Hlawka inequalities are tight in general: for any set $A$, some function $f$ attains the bound in \eqref{eqn:Koksma-Hlawka} up to constants. 
Consequently, a central goal in QMC methods has been to construct \emph{low-discrepancy} point sets with small $D_2^*(A)$. State-of-the-art constructions achieve $D_2^*(A) = \widetilde{O}_d(1)$ \cite{Nie92,DP10,Mat09,Owe13}, yielding the improved $n^{-1}$ convergence rate over MC:
\begin{equation}\label{eqn:QMCError}
    |\mathsf{err}(A,f)| \leq \widetilde{O}_d(\sigma_{\mathsf{HK}}(f) / n).
\end{equation}

\smallskip 
\noindent \textbf{Limitations of QMC.} 
Despite the quadratic improvement in convergence rate, QMC methods suffer from several severe limitations: (1) Their error bounds depend on the Hardy--Krause variation $\sigma_{\mathsf{HK}}(f)$, which can be much larger than $\sigma(f)$ and cause the QMC error to be much larger than MC for highly volatile integrands. For example, for $f(x) = \sin(kx)$, $\sigma_{\mathsf{HK}}(f) = \Theta(k)$ while $\sigma(f) = \Theta(1)$; (2) QMC methods require picking specific points, which may not be possible in the wide range of applications where only random samples are available---a feature that significantly limits their flexibility and applicability; (3) Using deterministic point sets may suffer from worst-case
outcomes and one cannot estimate the error statistically.

The distinctive strengths and weaknesses of MC and QMC methods raise a fundamental question: is there a method that can combine the benefits of both approaches, i.e., the flexibility of MC and the fast convergence rate of QMC? Progress toward this goal remained limited \cite{Lem09,Mat09,DP10,Owe13} until the recent work of Bansal and Jiang \cite{BJ25a}.

\smallskip 
\noindent \textbf{QMC Beyond Hardy--Krause.} Recently, Bansal and Jiang \cite{BJ25a} gave a randomized QMC method that naturally bridges the two approaches above and, more surprisingly, achieves significantly smaller error than the standard QMC benchmark of Koksma--Hlawka inequalities. 
In particular, they introduced a new notion called the {\em smoothed-out variation}, denoted by $\sigma_{\mathsf{SO}}(f)$, that is substantially smaller than $\sigma_{\mathsf{HK}}(f)$ (see \Cref{sec:prelims}), and proved the following improved bound:

\begin{restatable}[Bansal and Jiang \cite{BJ25a}]{theorem}{BansalJiangTheorem}
\label{thm:BJ25}
    There is a randomized algorithm that takes as input $n^2$ i.i.d. uniform samples from $[0,1)^d$ and in  $\widetilde{O}(n^2)$ time, partitions them into $n$ sets each of size $n$. Each one of these sets $A$ satisfies the following: for any function $f\in L^2([0,1)^{d})$ with finite $\sigma_{\mathsf{SO}}(f)$,
\[
\mathbb{E}\bigl[\mathsf{err}(A,f)^2\bigr]
\leq
\widetilde{O}_{d}(\sigma_{\mathsf{SO}}(f)^2 / n^2) .
\]
\end{restatable}

\smallskip
\noindent 
\textbf{Algorithm and Analysis in \cite{BJ25a}.} Inspired by the classical transference principle \cite{Mat09}, the  algorithm in \cite{BJ25a} operates roughly as follows. 
Starting with $n^2$ independent and uniform samples $A_0$ from $[0,1)^d$, it proceeds iteratively: at step $t$, for each one of the current sets $A_t$, it finds a balanced $\{\pm1\}$-coloring $x_t$  with low discrepancy and sub-Gaussianity with respect to {\em dyadic} boxes (see \Cref{fig:dyadicIntervals} for a 1-d illustration) of side length at least $1/\mathsf{poly}(n)$, and splits each $A_t$ into two subsets of equal size based on the two colors (see \Cref{fig:Transference}). 
Repeating this $\log n$ times partitions $A_0$ into $n$ sets of size $n$ each. 
We postpone a formal description of this algorithm to \Cref{subsec:alg}. 

\begin{figure}[ht]
    \centering
    \includegraphics[width=0.55\linewidth]{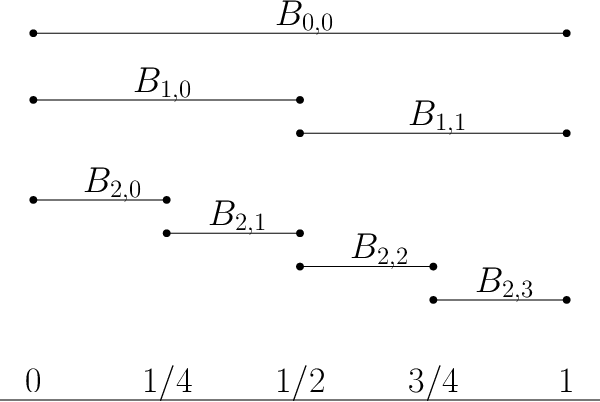}
    \caption{Dyadic Intervals}
    \label{fig:dyadicIntervals}
\end{figure}

The analysis in \cite{BJ25a}, however, is quite involved. We give more details in \Cref{subsec:ReviewBJ25}, but on a high level, they start with the Hlawka--Zaremba formula---an exact expression for $\err(A,f)$ as an integration of the  continuous discrepancy of $A$ multiplied by the mixed derivatives of $f$ (see \eqref{eqn:Hlawka--Zaremba} for the 1-d formula)---and discretize the formula along {\em prefix} boxes $[0,\mathbf{z}) := [0,z_1) \times \cdots \times [0,z_d)$ of granularity $1/\poly(n)$. Because the small sub-Gaussianity  of the coloring $\bx_t$ only holds for dyadic boxes rather than prefix boxes,\footnote{In fact, it was shown in \cite{BJ25a} that it is impossible to achieve small sub-Gaussianity for prefix boxes.} they need to relate prefix boxes (from the discretization) to dyadic boxes via dyadic decomposition, and carefully exploit its structure to obtain cancellations among the high-frequency components in the Fourier decomposition of $f$. These cancellations lead to the improvement in \Cref{thm:BJ25} over Koksma--Hlawka inequalities.

\smallskip
\noindent \textbf{Our Contribution.}  
The contribution of this article is twofold:
\begin{enumerate}
    \item Using standard ideas from functional analysis \cite{vybiral2006function,triebel2019function}, we give an equivalent characterization of $\sigma_{\mathsf{SO}}(f)$ in terms of the Haar–Besov seminorm of $f$, relating Bansal and Jiang's new notion of smoothed-out variation to classical quantities (see \Cref{thm:shift-averaged-mixed-haar-equivalence}).
    \item Through this characterization, we give a conceptually simpler and more direct analysis of the Bansal--Jiang QMC algorithm via Haar decomposition. As the Haar basis is more naturally aligned with dyadic boxes whose discrepancies the algorithm controls, this allows us to bypass various technicalities of \cite{BJ25a}: no need to use the Hlawka–Zaremba formula, Fourier decomposition, or the delicate cancellation arguments of \cite{BJ25a} that heavily exploit the structure of dyadic decomposition. See \Cref{subsecn:ourApproach} for an overview and \Cref{subsec:analysis} for details. 
\end{enumerate}

\begin{figure}[ht]
    \centering
    \includegraphics[width=0.7\linewidth]{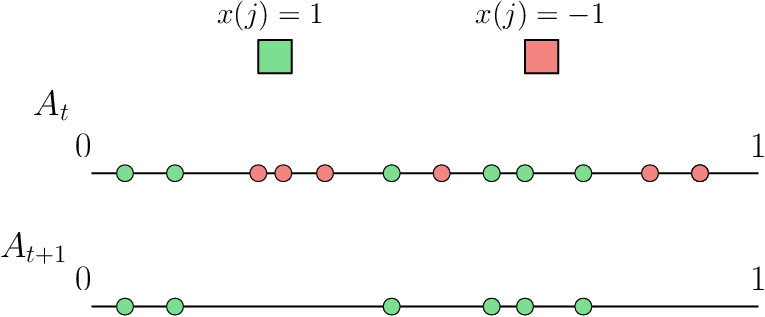}
    \caption{Iteration $t$ of the algorithm. $A_{t+1}$ is the subset of $A_t$ with color $1$.}
    \label{fig:Transference}
\end{figure}

\smallskip
\noindent \textbf{Roadmap.} In \Cref{secn:overview}, we review the proof strategy of Bansal and Jiang and identify the main sources of complexity in their analysis. We then outline our approach and explain how it simplifies the argument, both conceptually and technically. For completeness, we collect the notation used throughout the article in \Cref{sec:prelims}. We describe the algorithm of Bansal and Jiang \cite{BJ25a} in \Cref{subsec:alg}. In \Cref{subsec:analysis}, we present our simplified analysis of \Cref{thm:BJ25}. Finally, in \Cref{subsec:equivalence}, we prove the equivalence between the Haar--Besov seminorm and $\sigma_{\mathsf{SO}}(f)$.

\section{Overview}
\label{secn:overview}

In this section, we give a more detailed overview of the analysis in \cite{BJ25a}, followed by an overview of our approach that bypasses several of its technicalities.

\subsection{The Analysis of Bansal and Jiang}\label{subsec:ReviewBJ25}
For simplicity, we focus on the 1-d case and assume that $f$ is sufficiently smooth. 
Roughly speaking, the analysis in \cite{BJ25a} proceeds in the following three steps:

\smallskip
\noindent\textbf{Discretizing the Hlawka--Zaremba Formula.} 
Let $A$ be an arbitrary set output by the algorithm. 
They start with an exact expression for $\err(A,f)$ known as the Hlawka--Zaremba formula:  
\begin{equation}\label{eqn:Hlawka--Zaremba}
    \mathsf{err}(A,f)
    =
    \frac{1}{n}\int_0^1 f'(y)D(y)\,dy,
\end{equation}
where $D(z) := |A|z-\lvert A \cap[0,z)\rvert$ is the continuous discrepancy of the prefix interval $[0,z)$, i.e., the difference between the size of $[0,z)$ and the number of points from $A$ that lie in it.
Note that from \eqref{eqn:Hlawka--Zaremba}, one can actually obtain the Koksma--Hlawka inequality \eqref{eqn:Koksma-Hlawka} by applying Cauchy--Schwarz. Therefore, to obtain the improved bound in \Cref{thm:BJ25}, one must exploit {\em cancellations} in the integral in \eqref{eqn:Hlawka--Zaremba} which the Koksma--Hlawka inequalities completely give up.

Since the algorithm only controls the discrepancies of dyadic intervals of length at least $1/\poly(n)$, they need to discretize the integral in \eqref{eqn:Hlawka--Zaremba} into a sum with granularity $h = 1/\poly(n)$:
\begin{equation}\label{eqn:discretized-Hlawka--Zaremba}
\mathsf{err}(A,f) \approx \frac{h}{n} \sum_{j=1}^{1/h} f'(jh) D(jh) =: \frac{h}{n}\langle \mathbf{f}', \mathbf{D}_{\mathsf{prefix}}\rangle,
\end{equation}
where $\mathbf{f}'$ and $\mathbf{D}_{\mathsf{prefix}}$ denote the vectors formed by all the $f'(jh)$ and $D(jh)$. 
The discretization error in \eqref{eqn:discretized-Hlawka--Zaremba} is roughly on the order of $O(h) = 1/\poly(n)$, and hence can be safely ignored.

\smallskip
\noindent\textbf{From Prefix to Dyadic Discrepancies.} 
Bansal and Jiang's crucial insight is that cancellations in \eqref{eqn:discretized-Hlawka--Zaremba} should come from the sub-Gaussianity of the algorithm. However, there is a mismatch here---the algorithm only guarantees low sub-Gaussianity for the discrepancies of dyadic intervals, while the vector $\mathbf{D}_{\mathsf{prefix}}$ in \eqref{eqn:discretized-Hlawka--Zaremba} corresponds to the discrepancies of prefix intervals. 
More problematically, it was shown in \cite{BJ25a} that no algorithm can attain low sub-Gaussianity for $\mathbf{D}_{\mathsf{prefix}}$.

To bypass this failure of sub-Gaussianity for prefix intervals $\mathbf{D}_{\mathsf{prefix}}$, Bansal and Jiang expressed $\mathbf{D}_{\mathsf{prefix}}$ in terms of the discrepancies of dyadic intervals $\mathbf{D}_{\mathsf{dyadic}}$ via the dyadic decomposition matrix $P$ (with granularity $h =1/\poly(n)$) as $\mathbf{D}_{\mathsf{prefix}} = P \mathbf{D}_{\mathsf{dyadic}}$. Then \eqref{eqn:discretized-Hlawka--Zaremba} gives 
\begin{equation}\label{eqn:discretized-Hlawka--Zaremba2}
\mathsf{err}(A,f) \approx \frac{h}{n}\langle \mathbf{f}', \mathbf{D}_{\mathsf{prefix}}\rangle = \frac{h}{n} \langle P^\top \mathbf{f}', \mathbf{D}_{\mathsf{dyadic}} \rangle \approx  \frac{h}{n} \cdot \|P^\top \mathbf{f}'\|_2 ,
\end{equation}
where the last step uses that $\mathbf{D}_{\mathsf{dyadic}}$ is $\widetilde{O}(1)$ sub-Gaussian, as guaranteed by the algorithm.

\smallskip
\noindent\textbf{Cancellations for High-Frequency Fourier Components.}  Finally, to bound $\|P^\top \mathbf{f}'\|_2$ in terms of $\sigma_{\mathsf{SO}}(f)$, they use the Fourier decomposition of $f$ and exploit delicate structures of the dyadic decomposition matrix $P$ to find cancellations in the high-frequency Fourier components of $f$.

\subsection{Our Approach}\label{subsecn:ourApproach}

Note that the various technicalities in the \cite{BJ25a} analysis all result from the mismatch between the discrepancies of dyadic intervals that the algorithm controls and those of prefix intervals that the analysis requires in order to find cancellations in the Hlawka--Zaremba formula \eqref{eqn:Hlawka--Zaremba}. 
To avoid these technicalities, we abandon the Hlawka--Zaremba formula entirely, and directly analyze $\err(A,f)$ via the Haar decomposition of $f$. 
As the Haar basis aligns naturally with dyadic intervals, this leads to a conceptually simpler and more direct analysis of the Bansal--Jiang algorithm. 

In particular, let $f = \sum_{j,\ell} \langle f, \psi_{j,\ell} \rangle \psi_{j,\ell}$ be its Haar decomposition, where $\psi_{j,\ell}$ denote the normalized Haar basis functions (see \Cref{sec:prelims} and \Cref{fig:HaarFunc}). One can then write 
\begin{align*}
    \err(A,f) =\sum_{j,\ell} \langle f, \psi_{j,\ell} \rangle \cdot \err(A,\psi_{j,\ell}) .
\end{align*}
As the discrepancies of dyadic intervals are $\widetilde{O}(1)$-sub-Gaussian, the same guarantee holds for the $\err(A,\psi_{j,\ell})$ (up to a normalization factor in $\psi_{j,\ell}$). It then follows that
\[
\err(A,f)^{2} \approx \frac{1}{n^{2}}\sum_{j,\ell} 2^j \langle f, \psi_{j,\ell} \rangle^2 , 
\]
where the quantity on the right-hand side above is exactly the square of the Haar--Besov seminorm $\|f\|_{\mathsf{H}}$ (see \Cref{defn:HaarBesovSeminorm}).
Finally, using standard ideas from the functional analysis literature (e.g., see \cite{vybiral2006function,triebel2019function}), we show that (a slight variant) of $\|f\|_{\mathsf H}$ is in fact equivalent to the smoothed-out variation $\sigma_{\mathsf{SO}}(f)$, giving a more classical interpretation of this new notion in \cite{BJ25a}.

\begin{figure}[ht]
    \centering
    \includegraphics[width=0.55\linewidth]{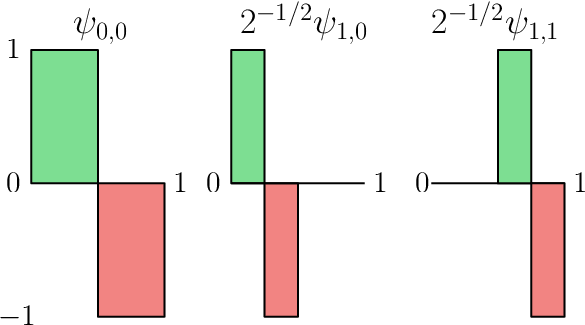}
    \caption{Unnormalized Haar functions $2^{-j/2}\psi_{j,\ell}$ for $j = 0,1$.}
    \label{fig:HaarFunc}
\end{figure}

\section{Preliminaries}\label{sec:prelims}
\smallskip
\noindent \textbf{Basic Notation.} In this section, we formally define the terminology used throughout the article. For integers $j\geq 0$ and $0\leq \ell<2^j$, let $B_{j,\ell}(z)$ denote the indicator of the $\ell^{\text{th}}$ dyadic interval at level $j$ in $[0,1)$. In particular, $B_{0,0}(y)$ denotes the constant function $1$ on $[0,1)$.
We use $\mathcal{D}(h)$ to denote the collection of all dyadic intervals $B_{j,\ell}$
for $0 \leq j \leq h$ and $0 \leq \ell < 2^{j}$.
Likewise, let $\psi_{j,\ell}$ denote the $(j,\ell)^{\text{th}}$ normalized Haar function, defined by $\psi_{j,\ell}(z)=2^{j/2}\bigl(B_{j+1,2\ell}(z)-B_{j+1,2\ell+1}(z)\bigr)$. Together with the constant function $1$ on $[0,1)$, the Haar functions form an orthonormal basis for $L^2([0,1))$, the space of square-integrable functions equipped with the inner product $\langle f,g\rangle=\int_0^1 f(z)g(z)\,dz$. 

\smallskip
\noindent \textbf{Notation for Higher Dimension.} For an integer \(d>0\), we denote
\([d]=\{1,\ldots,d\}\). For sets $S,T \subseteq [d]$, we use $\bar{S}$ to denote $[d]\setminus S$ and $S \setminus T$ to denote $S \cap \bar{T}$.
We use bold letters, e.g., \(\mathbf j\), to denote vectors, and
\(\mathbf j_i\) to denote the \(i\)th coordinate of \(\mathbf j\). For $S\subseteq[d]$, let $\mathbf{j}$ and $\boldsymbol{\ell}$ be tuples indexed by $S$, where $\mathbf{j}_i$ is a nonnegative integer and $0\leq\boldsymbol{\ell}_i<2^{\mathbf{j}_i}$ for every $i\in S$. We denote $B_{S,\mathbf{j},\boldsymbol{\ell}}(\bz) := \prod_{i \in S}B_{\mathbf{j}_{i},\boldsymbol{\ell}_{i}}(\bz_{i})$ the dyadic box and write $\mathcal{D}^d(h)$ for the collection of dyadic boxes indexed by nonempty subsets $S\subseteq[d]$, multi-indices $\mathbf{j}$ satisfying $\max_{i\in S}\mathbf{j}_i\leq h$, and all corresponding values of $\boldsymbol{\ell}$. Likewise, we write $\Psi_{S,\mathbf{j},\boldsymbol{\ell}}$ for the corresponding Haar function, defined by $\Psi_{S,\mathbf{j},\boldsymbol{\ell}}(\bz)=\prod_{i\in S}\psi_{\mathbf{j}_i,\boldsymbol{\ell}_i}(\bz_i)$ for $\bz \in [0,1)^{d}$. 

\smallskip
\noindent \textbf{Multi-Dimension Haar Functions.} As in one dimension, the constant function $1$, together with the Haar functions $\Psi_{S,\mathbf{j},\boldsymbol{\ell}}$ over all nonempty subsets $S\subseteq[d]$ and all corresponding tuples $\mathbf{j},\boldsymbol{\ell}$, forms an orthonormal basis for $L^2([0,1)^d)$ with respect to the inner product $\langle f,g\rangle=\int_{[0,1)^d}f(\bz)g(\bz)\,d\bz$. To familiarize the reader with multi-dimensional Haar functions, we briefly describe them as sums over the corresponding dyadic boxes.
\begin{align}
    \Psi_{S, \mathbf{j},\boldsymbol{\ell}}(\bz) &= \displaystyle\prod_{i \in S}\psi_{\mathbf{j}_{i},\boldsymbol{\ell}_{i}}(\bz_{i}) = \displaystyle\prod_{i\in S}2^{\mathbf{j}_{i}/2}(B_{\mathbf{j}_{i} + 1,2\boldsymbol{\ell}_{i}}(\bz_{i}) - B_{\mathbf{j}_{i} + 1,2\boldsymbol{\ell}_{i}+1}(\bz_{i})) \nonumber \\
    &= 2^{|\mathbf{j}|/2}\sum_{T \subseteq S}(-1)^{|S\setminus T|}\displaystyle\prod_{i \in T}B_{\mathbf{j}_{i} + 1,2\boldsymbol{\ell}_{i}}(\bz_{i})\displaystyle\prod_{i\in S \setminus T}B_{\mathbf{j}_{i} + 1,\boldsymbol{2\ell_{i}+1}}(\bz_{i}) \label{eqn:familiaritywithHaareqn3}\\
    &=2^{|\mathbf{j}|/2}\sum_{T \subseteq S}(-1)^{|S\setminus T|}B_{S,\mathbf{j} + \mathbf{1},2\boldsymbol{\ell} + \mathbf{1}_{S\setminus T}}(\bz) \LinesNotNumbered, 
\end{align}

where in \eqref{eqn:familiaritywithHaareqn3}, we write $|\mathbf{j}|=\sum_{i\in S}|\mathbf{j}_i|$, and $\mathbf{1}$ denotes the all-ones tuple. Similarly, $\mathbf{1}_{S\setminus T}$ denotes the tuple that equals $1$ on the indices in $S\setminus T$ and $0$ elsewhere, and tuple addition is performed coordinate-wise.

\smallskip
\noindent \textbf{Haar--Besov Seminorm.} The Haar--Besov seminorm is central to our analysis. We refer the interested reader to \cite{vybiral2006function,triebel2019function} for background.
\begin{definition}[Haar--Besov Seminorm]\label{defn:HaarBesovSeminorm}
The Haar--Besov seminorm of $f\in L^2([0,1)^d)$ is
\[
    \lVert f\rVert_{\mathsf H}^2
    :=
    \sum_{\substack{S\subseteq[d]\\S\neq\emptyset}}
    \sum_{\mathbf{j},\boldsymbol{\ell}}
    2^{|\mathbf{j}|}
    |\langle f,\Psi_{S,\mathbf{j},\boldsymbol{\ell}}\rangle|^2,
\]
where, for each $S$, the tuples $\mathbf{j}$ and $\boldsymbol{\ell}$ are indexed by $S$, with $\mathbf{j}_i$ a nonnegative integer and $0\leq\boldsymbol{\ell}_i<2^{\mathbf{j}_i}$ for every $i\in S$. Throughout, $\mathbf{j},\boldsymbol{\ell}$ range over all tuples satisfying these constraints unless specified.
\end{definition}

\smallskip
\noindent \textbf{Fourier Analysis on the Unit Cube.} Finally, we introduce notation for Fourier analysis on $[0,1)^d$, which we use to prove the equivalence between $\sigma_{\mathsf{SO}}(f)$ and $\lVert f\rVert_{\mathsf H}$ in \Cref{subsec:equivalence}. Let $f: [0,1)^{d} \rightarrow \mathbb{C}$ denote a $1$-periodic function in $L^{2}([0,1)^{d})$. For $\mathbf{k} \in \mathbb{Z}^{d}$, we use $\widehat{f}_{\mathbf{k}}$ to denote the Fourier coefficient $\int_{[0,1)^{d}}f(\bz)\exp(-2\pi i \langle \mathbf{k}, \bz \rangle)d\bz$ and the Fourier series of $f$ is given by,
\[
\sum_{\mathbf{k} \in \mathbb{Z}^{d}}\widehat{f}_{\mathbf{k}}\exp(2\pi i \langle \mathbf{k},\bz\rangle ).
\]

We use $e_{\mathbf{k}}$ as a shorthand to denote $\exp(2\pi i \langle \mathbf{k}, \bz \rangle)$.
The Fourier characters $e_{\mathbf{k}}$, for $\mathbf{k}\in\mathbb{Z}^d$, form an orthonormal basis for complex-valued functions in $L^2([0,1)^d)$ with respect to the inner product $\langle f,g\rangle=\int_{[0,1)^d}f(\bz)\overline{g(\bz)}\,d\bz$. Consequently, Parseval's identity gives $\int_{[0,1)^d}|f(\bz)|^2\,d\bz=\sum_{\mathbf{k}\in\mathbb{Z}^d}|\widehat{f}_{\mathbf{k}}|^2$. Now that we are equipped with the notation for Fourier analysis, we can now define:
\[
\sigma_{\mathsf{SO}}(f)^2 = \sum_{\substack{\mathbf{k}\in\mathbb{Z}^d,\\ \mathbf{k}\neq\mathbf{0}}} |\widehat{f}_{\mathbf{k}}|^{2}\prod_{i \in [d]}\max(1, |\mathbf{k}_{i}|), \label{eq:smoothed-out-variation} \quad
\sigma_{\mathsf{HK}}(f)^2 = \sum_{\substack{\mathbf{k}\in\mathbb{Z}^d,\\ \mathbf{k}\neq\mathbf{0}}} |\widehat{f}_{\mathbf{k}}|^{2}\prod_{i \in [d]}\max(1, |\mathbf{k}_{i}|^{2}),
\]
where $\sigma_{\mathsf{HK}}(f)$ denotes the Hardy--Krause variation of $f$ (see \cite{BJ25a} for more details).

\section{Algorithm and Analysis}

In \Cref{subsec:alg}, we describe the algorithm of Bansal and Jiang \cite{BJ25a} and state the guarantees that we need. We then analyze its error in \Cref{subsec:analysis}.

\subsection{Algorithm}\label{subsec:alg}

\smallskip
\noindent \textbf{The Setup.} The algorithm of Bansal and Jiang \cite{BJ25a} uses a key ingredient: the self-balancing walk of \cite{ALS21}. While \cite{BJ25a} use the self-balancing walk in \cite{ALS21}, as the recent work \cite{AA26} supersedes \cite{ALS21} we use that instead.

 Given a collection of vectors, the self-balancing walk efficiently signs vectors so that the resulting signed sum has small sub-Gaussian\footnote{A random vector $\mathbf{r}$ is $\sigma$-sub-Gaussian if, for every unit vector $\theta$ and every $t \geq 0$, $\Pr\!\left[\langle \mathbf{r}, \theta \rangle \geq t\sigma\right]
    \leq \exp\!\left(-t^{2}/2\right).$ } norm. Formally,
\begin{theorem}[Gaussian Triplet Walk of \cite{AA26}]
\label{thm:selfBalancing}
There exists an algorithm $\mathsf{TripletWalk}$ such that, given vectors $v_1,v_2,\ldots,v_n$ with $\ell_{2}$-norm at most $1$, it computes a symmetric and balanced \footnote{A signing is symmetric if

it is symmetrically distributed, and balanced if exactly half of the entries are $+1$.}random signing $x\in\{\pm1\}^n$ such that $\sum_{j\leq t}x(j)v_j$ is $O(1)$ sub-Gaussian for every $t\in[n]$ with high probability. The algorithm runs in time $O\bigl(\sum_{j\in[n]}\operatorname{nnz}(v_j)\bigr)$.
\end{theorem}

The algorithm proceeds iteratively by applying \Cref{thm:selfBalancing} to a suitable collection of vectors at each step. We introduce notation for these vectors to simplify the description of the algorithm. Recall that $\mathcal{D}^{d}(h)$ denotes the collection of all dyadic boxes whose level in each coordinate is at most $h$ (see \Cref{sec:prelims}). Given a set $A \subseteq [0,1)^{d}$ and a point $\bz \in A$, let $v_{\bz} \in \mathbb{R}^{A+\mathcal{D}^{d}(h)}$ be defined as follows: $(i)$ the coordinates indexed by $A$ indicate the point $\bz$, and $(ii)$ the coordinates indexed by $\mathcal{D}^{d}(h)$ indicate which dyadic boxes contain $\bz$.

\subsubsection{Algorithm Description}\label{subsubsecn:algDescription}
The algorithm of \cite{BJ25a} proceeds as follows:
Set $h = O(\log n)$ with a large constant. Sample $\bs \sim \text{Unif}([0,1)^{d})$ and shift $[0,1)^{d}$ by $\bs$. 
\begin{enumerate}
    \item Sample $n^2$ independent uniform points from $[0,1)^{d}$, and let $A_0$ be the resulting collection.

    \item For each step $t=0,1,\ldots,T-1$, where $T=\log n$, construct the vector $v_{\bz}$ for every $\bz \in A_{t}$ using the representation above. Apply $\vectorBal$ to these vectors to obtain a balanced signing $x_t$, and let $A_{t+1}$ be the collection of points $\bz \in A_t$ for which $x_t(\bz)=1$.

    \item Output the average of $f$ on $A_{T}$:
    \[
    \frac{1}{|A_{T}|}\sum_{\bz \in A_{T}}f(\bz).
    \]
\end{enumerate}

We next describe an important property of Algorithm \ref{subsubsecn:algDescription}: A direct application of \Cref{thm:selfBalancing} shows that for each step $t \geq 0$, $\sum_{z \in A_t}x_t(\bz)v_{\bz}$ is $O(\log^{d/2}n)$ sub-Gaussian. This statement is proved in the same form as Lemma 3.3 of Bansal and Jiang \cite{BJ25a}, so we omit the proof.

\begin{proposition}\label{prop:sub-Gaussianity}
    At every step $t$, $\sum_{z \in A_{t}}x_{t}(\bz)v_{\bz}$ is $O(\log^{d/2}n)$ sub-Gaussian.
\end{proposition}

\subsection{Analysis}\label{subsec:analysis}

In this section, we give our proof for \Cref{thm:BJ25}. We restate it for completeness.
\BansalJiangTheorem*

\smallskip
\noindent\textbf{Roadmap.} Following the approach outlined in \Cref{subsecn:ourApproach}, our proof has two main ingredients. First, we analyze $\mathsf{err}(A,f)$ directly in the Haar basis and bound it in terms of $\lVert f\rVert_{\mathsf{H}}$. Since the algorithm controls dyadic discrepancies only up to a fixed level, we treat the low- and high-frequency components of $f$ separately in \Cref{subsubsec:LowFrequency} and \Cref{subsubsec:highfrequencyanalysis}. Second, we use the equivalence between the shift-averaged Haar seminorm and $\sigma_{\mathsf{SO}}(f)$, which we prove in \Cref{subsec:equivalence}. We combine these ingredients in \Cref{subsubsec:finalProof}.

\subsubsection{Proof of \Cref{thm:BJ25}}\label{subsubsec:finalProof}
Recall from \Cref{subsec:alg} that the algorithm outputs a set $A_T$ of $n$ points and uses $\bar{f}(A_T)$ as an estimate of $\bar{f}$.
Our goal is to bound $\mathsf{err}(A_{T}, f)$ in terms of the smoothed-out variation $\sigma_{\mathsf{SO}}(f)$.

\smallskip
\noindent\textbf{Low- and High-Frequency Functions.} 
Recall that the algorithm in \Cref{subsubsecn:algDescription} controls discrepancies over dyadic boxes in $\mathcal{D}^{d}(h)$, where $h=O(\log n)$. For this choice of $h$, let $P_{<h}$ denote the orthogonal projector onto the subspace of $L^2([0,1)^d)$ spanned by all Haar functions $\Psi_{S,\mathbf{j},\boldsymbol{\ell}}$ indexed by $S\subseteq[d]$, with $0\leq \mathbf{j}_i<h$ and $0\leq \boldsymbol{\ell}_i<2^{\mathbf{j}_i}$ for each $i\in S$.
We also let $P_{\geq h}=I-P_{< h}$ denote the projection onto the orthogonal complement of this subspace, where $I$ is the identity operator on $L^2([0,1)^{d})$.

Recall that Algorithm \ref{subsec:alg} applies a uniformly random shift $\bs$ in $[0,1)^{d}$ to the dyadic boxes. Shifting the dyadic boxes by $\bs$ is equivalent to keeping the boxes fixed and replacing $f$ by $f_{\bs}$, where $f_{\bs}(\bz) = f((\bz + \bs) \bmod 1)$. In \Cref{lem:lowfrequencyError} and \Cref{lem:highFrequency}, we show that for every shift $\bs \in [0,1)^{d} $:
\begin{align}
    \mathbb{E}_{\nu_{\bs}}[\mathsf{err}(A_{T}, P_{<h}f_{\bs})^{2}] \leq O_{d}(\log^{d} n) \cdot \frac{\lVert P_{<h}f_{\bs}\rVert_{\mathsf H}^{2}}{n^{2}}, \label{eqn:guaranteeComponents1}\\
    \mathbb{E}_{\nu_{\bs}}[\mathsf{err}(A_{T}, P_{\geq h}f_{\bs})^{2}] \leq O_{d}(\log^{d} n) \cdot \frac{\lVert P_{\geq h}f_{\bs}\rVert_{\mathsf H}^{2}}{\poly(n)} \label{eqn:guaranteeComponents2},
\end{align}
where $\nu_{\bs}$ denotes the distribution of the point set output by the algorithm for a fixed shift $\bs$ and $O_{d}(\cdot)$ hides factors depending only on $d$.
Here, we remark that \eqref{eqn:guaranteeComponents1} is the crux of the proof, while \eqref{eqn:guaranteeComponents2} is a ``rounding'' error that can be safely ignored. 

\smallskip
\noindent\textbf{Equivalence Characterization.} Our proof relies crucially on the equivalence between the Haar--Besov seminorm averaged over random shifts and $\sigma_{\mathsf{SO}}(f)$. This equivalence follows from standard ideas in functional analysis \cite{vybiral2006function,triebel2019function}. We give a full proof in \Cref{subsec:equivalence}.

\begin{restatable}[Shift-averaged seminorm equivalence]
    {theorem}{shiftAveragedHaarEquivalence}
\label{thm:shift-averaged-mixed-haar-equivalence}
For any \(f\in L^2([0,1)^d)\), its average Haar--Besov seminorm over random shifts $\bs \sim \mathrm{Unif}([0,1)^{d})$ is equivalent to its smoothed-out variation:
\[
    \mathbb{E}_{\mathbf{s}\sim [0,1)^{d}}
    [\lVert f_{\mathbf{s}}\rVert_{\mathsf H}^{2}]
    = 
    \Theta_{d}(\sigma_{\mathsf{SO}}(f)^2),
\]
where $\sigma_{\mathsf{SO}}(f)^2 = \sum_{\substack{\mathbf{k}\in\mathbb{Z}^d,\\ \mathbf{k}\neq\mathbf{0}}} |\widehat{f}_{\mathbf{k}}|^{2}\prod_{i \in [d]}\max(1, |\mathbf{k}_{i}|)$.
\end{restatable}

Without loss of generality\footnote{Given a function $f$, we can analyze $f - \bar{f}$. This leaves the integration error unchanged and, because the signing is balanced, also leaves the weighted discrepancy unchanged.}, assume that $f \in L^2([0,1)^{d})$ has mean zero: $\bar{f} = 0$. Combining \eqref{eqn:guaranteeComponents1} and \eqref{eqn:guaranteeComponents2} with \Cref{thm:shift-averaged-mixed-haar-equivalence} completes the proof.
\begin{proof}[Proof of \Cref{thm:BJ25}]
    Since $\mathsf{err}(A_T,\,\cdot\,)$ is linear, we decompose the error into the contributions from the low and high-frequency Haar components of $f_{\bs}$:
    \begin{align*}\label{eqn:LinearityOfError}
    \mathsf{err}(A_{T}, f_{\bs}) = \frac{1}{|A_{T}|}\sum_{z \in A_{T}}f_{\bs}(\bz)
    &= \frac{1}{|A_{T}|}\sum_{z \in A_{T}}((P_{< h}f_{\bs})(\bz) + P_{\geq h}f_{\bs}(\bz))\\ &= \mathsf{err}(A_{T}, P_{< h}f_{\bs}) + \mathsf{err}(A_{T}, P_{\geq h}f_{\bs}).
    \end{align*}
Using \Cref{lem:lowfrequencyError} (see \Cref{eqn:guaranteeComponents1}), together with the equivalence between the Haar--Besov seminorm and the smoothed-out variation (\Cref{thm:shift-averaged-mixed-haar-equivalence}), we obtain:
\begin{align*}
   \mathbb{E}_{s}\bigl[\,\mathbb{E}_{\nu_{\bs}}[\mathsf{err}(A_{T}, P_{<h}f_{\bs} )^{2}]\,\bigr] &\leq \mathbb{E}_{s}[\lVert P_{<h}f_{\bs}\rVert_{\mathsf H}^{2}]\frac{O_{d}(\log^{d}n)}{n^{2}}\\
   &\leq \frac{O_{d}(\log^{d}n)}{n^{2}}\mathbb{E}_{s}[\lVert f_{\bs}\rVert_{\mathsf H}^{2}] \leq  \frac{O_{d}(\log^{d}n)}{n^{2}}\sigma_{\mathsf{SO}}(f)^2.
\end{align*}
The bound on $\mathsf{err}(A_{T}, P_{\geq h}f_{\bs})$ holds analogous to above via \Cref{lem:highFrequency}. We bound $\mathsf{err}(A_{T},f)$:
\begin{align*}
    \mathbb{E}_{\nu}[\mathsf{err}(A_{T}, f)^{2}] &= \mathbb{E}_{\bs}[\,\mathbb{E}_{\nu_{\bs}}[\mathsf{err}(A_{T}, f_{\bs})^{2}]\,]\\
    &\lesssim \mathbb{E}_{\bs}[\,\mathbb{E}_{\nu_{\bs}}[\mathsf{err}(A_{T}, P_{< h}f_{\bs})^{2}]\,] + \mathbb{E}_{\bs}[\,\mathbb{E}_{\nu_{\bs}}[\mathsf{err}(A_{T}, P_{\geq h}f_{\bs})^{2}]\,]
    \leq \frac{O_{d}(\log^{d}n)}{n^{2}}\sigma_{\mathsf{SO}}(f)^2,
\end{align*}
where $\nu$ denotes the distribution of the output point set on average over $\bs$.
\end{proof}

\subsubsection{Bounding Error on the Low-Frequency Components}\label{subsubsec:LowFrequency}

In this section, we bound the error incurred by Algorithm \ref{subsubsecn:algDescription} on the low-frequency Haar projection of $f$. Recall from \Cref{subsubsec:finalProof} that $P_{< h}$ denotes the orthogonal projector onto the subspace of $L^2([0,1)^d)$ spanned by all Haar functions $\Psi_{S,\mathbf{j},\boldsymbol{\ell}}$ indexed by $S\subseteq[d]$, with $0\leq \mathbf{j}_i<h$ and $0\leq \boldsymbol{\ell}_i<2^{\mathbf{j}_i}$ for each $i\in S$. For $\bz \in [0,1)^{d}$,
\begingroup
\setlength{\belowdisplayskip}{2pt}
\setlength{\belowdisplayshortskip}{1pt}
\[
(P_{< h}f)(\bz) = \sum_{\substack{S \subseteq [d], \\ S \neq \emptyset}}\sum_{\substack{\mathbf{j},\boldsymbol{\ell} \in P_{<h}^{S}}}\langle \Psi_{S,\mathbf{j},\boldsymbol{\ell}}, f\rangle \,\Psi_{S,\mathbf{j},\boldsymbol{\ell}}(\bz),
\]
where $\mathbf{j},\boldsymbol{\ell}\in P_{<h}^S$ indicates that the indices range over all Haar functions associated with a nonempty set $S\subseteq[d]$ and satisfying the constraints above. In this section, we show the following.
\endgroup
\begin{lemma}\label{lem:lowfrequencyError}
For a function $f \in L^2([0,1)^{d})$ and a fixed shift $\bs \in [0,1)^{d}$, the second moment of the low-frequency error of Algorithm \ref{subsubsecn:algDescription} satisfies: 
    \[
    \mathbb{E}_{\nu_{\bs}}\,[ \err(A_{T},P_{< h}f_{\bs})^{2}] \leq O_{d}(\log^{d} n)\frac{\lVert P_{< h}f_{\bs}\rVert_{\mathsf H}^{2}}{n^{2}},
    \]
    where $\nu_{\bs}$ denotes the distribution of the output point set conditioned on the shift $\bs$ and $O_{d}(\cdot)$ hides factors depending only on $d$.
\end{lemma}

We prove the error bound in two steps: $(i)$ we express $\mathsf{err}(A_T,P_{<h}f_{\bs})$ in terms of the per-step discrepancies incurred by Algorithm \ref{subsubsecn:algDescription}, and $(ii)$ we use the sub-Gaussianity of dyadic-box discrepancies (see \Cref{prop:sub-Gaussianity}) to bound the resulting expression.

\smallskip
\noindent\textbf{Expressing Error Iteratively.} Recall that the algorithm in \Cref{subsubsecn:algDescription} iteratively \emph{thins} the point set by applying \Cref{thm:selfBalancing} to a suitable collection of vectors and retaining the points colored $+1$. We express $\mathsf{err}(A_T,f_{\bs})$ in terms of the per-step discrepancies controlled by \Cref{thm:selfBalancing}. To do so, we first define a notion of \emph{weighted} discrepancy of a function:
\begin{definition}[Weighted Discrepancy]\label{defn:ErrorWeightedDisc}
Given a function $f \in L^2([0,1)^{d})$, we let its weighted discrepancy at step $t$ be,
\[
    \mathsf{disc}_t(f)
    =
    \sum_{\bz \in A_t}x_t(\bz)f(\bz).
\]
\end{definition}
The following is analogous to Lemma 3.4 of \cite{BJ25a}. As in \Cref{subsubsec:finalProof}, we assume throughout that $f$ has mean zero.

\begin{proposition}
\label{prop:TransferenceErrorBound}
    For any function $g \in L^{2}([0,1)^{d})$, the error incurred by Algorithm in estimating $\bar{g}$ in Algorithm \ref{subsubsecn:algDescription} is:
        \[
    \mathsf{err}(A_{T},g) = \mathsf{err}(A_{0},g) + \sum_{t < T} \frac{\mathsf{disc}_{t}(g)}{|A_{t}|}.
    \]
\end{proposition}

\begin{proof}[Proof of \Cref{prop:TransferenceErrorBound}]
We begin by recursively expressing $\mathsf{err}(A_{t},g )$ in terms of $\mathsf{err}(A_{t-1},g)$ and $\mathsf{disc}_{t-1}(g)$ for $t \geq 1$:
\[
\mathsf{err}(A_{t},g) = \frac{1}{|A_t|}\sum_{\bz \in A_t} g(\bz)
= \frac{1}{|A_{t-1}|}\sum_{\bz \in A_{t-1}} \bigl(1+x_{t-1}(\bz)\bigr)g(\bz),
\]
where the equality follows by adding and subtracting the contributions of the negatively signed points and using the fact that the signing is balanced. Consequently, 
\begin{align}\label{eqn:OneStepIterationTransference}
\frac{1}{|A_{t-1}|}\sum_{\bz \in A_{t-1}} \bigl(1+x_{t-1}(\bz)\bigr)g(\bz) 
&= \mathsf{err}(A_{t-1},g)
   + \frac{1}{|A_{t-1}|}\sum_{\bz \in A_{t-1}} x_{t-1}(\bz)g(\bz)\\
   &=\mathsf{err}(A_{t-1},g) + \frac{\mathsf{disc}_{t-1}(g)}{|A_{t-1}|}.
\end{align}

Iterating over all the steps completes the proof.
\end{proof}

To bound the second moment of $\mathsf{err}(A_T,P_{<h}f_{\bs})$, we express it in \Cref{prop:SecondMomentBoundsOnError} as a sum of the second moments of the per-step discrepancies $\mathsf{disc}_t(P_{<h}f_{\bs})$. Crucially, the symmetry of the signings $x_t$ produced by \Cref{thm:selfBalancing} implies that, conditioned on the first $t-1$ iterations, the discrepancy at step $t$ has expectation zero. Consequently, all cross terms between discrepancies from different steps vanish.

\begin{proposition}\label{prop:SecondMomentBoundsOnError}
    For any function $g \in L^{2}([0,1)^{d})$, the second moment of the error incurred by Algorithm in estimating $\bar{g}$ in Algorithm \ref{subsubsecn:algDescription} is:
        \[
    \mathbb{E}_{\nu_{\bs}}[\mathsf{err}(A_{T},g)^{2}] = \mathbb{E}_{\nu_{\bs}}[\mathsf{err}(A_{0},g)^{2}] + \sum_{t < T} \frac{\mathbb{E}_{\nu_{\bs}}[\mathsf{disc}_{t}(g)^{2}]}{|A_{t}|^{2}},
    \]
    where $\nu_{\bs}$ denotes the distribution of the point set output by the algorithm for $\bs$.
\end{proposition}

\begin{proof}[Proof of \Cref{prop:SecondMomentBoundsOnError}]
   Recall from \eqref{eqn:OneStepIterationTransference} in the proof of \Cref{prop:TransferenceErrorBound} that, for every $t\geq 1$,
$\mathsf{err}(A_t,g)$ is the sum of $\mathsf{err}(A_{t-1},g)$ and $\mathsf{disc}_{t-1}(g)/|A_{t-1}|$. In what follows, $\nu_{\bs,t-1}$ denotes the distribution $\nu_{\bs}$ conditioned on the first $t-1$ iterations:
\begin{align*}
    \mathbb{E}_{\nu_{\bs},t-1}[\mathsf{err}(A_t,g)^{2}] &= \mathbb{E}_{\nu_{s},t-1}[\bigl(\mathsf{err}(A_{t-1},g)+\frac{\mathsf{disc}_{t-1}(g)}{|A_{t-1}|}\bigr)^{2}]\\
    &= \mathsf{err}(A_{t-1},g)^{2} + \frac{\mathbb{E}_{\nu_{\bs},t-1}[\mathsf{disc}_{t-1}(g)^{2}]}{|A_{t-1}|^{2}} +\frac{2\mathsf{err}(A_{t-1},g)\mathbb{E}_{\nu_{\bs},t-1}[\mathsf{disc}_{t-1}(g)]}{|A_{t-1}|}\\
    &= \mathsf{err}(A_{t-1},g)^{2} + \frac{\mathbb{E}_{\nu_{\bs},t-1}[\mathsf{disc}_{t-1}(g)^{2}]}{|A_{t-1}|^{2}},
\end{align*}
where the second equality follows as $\mathbb{E}_{\nu_{\bs},t-1}[\mathsf{disc}_{t-1}(g)] = 0$ by symmetry of $x_{t-1}$. Iterating for $t=0,1,\ldots,T-1$ completes the proof.
\end{proof}

We bound the \emph{weighted discrepancy} $\mathsf{disc}_t(P_{<h}f_{\bs})$ via the sub-Gaussian control of Algorithm \ref{subsubsecn:algDescription} over dyadic boxes. 

\smallskip
\noindent\textbf{Bounding the Per-Step Error via sub-Gaussianity.} We now turn to the key argument underlying our analysis.
Recall that, at each step $t$, the algorithm applies \Cref{thm:selfBalancing} to obtain a signing $x_t$ and the corresponding discrepancy vector $\sum_{\bz\in A_t}x_t(\bz)v_{\bz}$. Its coordinates indexed by $A_t$ and $\mathcal{D}^{d}(h)$ indicate the colors of the points and the discrepancies over dyadic boxes respectively. In the penultimate step of the proof of \Cref{lem:lowfrequencyError}, we analyze $\mathsf{disc}_t(P_{<h}f_{\bs})$ in \Cref{lem:low-frequencyDiscrepancyBound}. We expand $f_{\bs}$ in the Haar basis and express the discrepancy as an inner product between the discrepancy vector and the Haar coefficients of $f_{\bs}$. We then use the sub-Gaussianity of the discrepancy vector (see \Cref{prop:sub-Gaussianity}) to bound this inner product.

\begin{lemma}\label{lem:low-frequencyDiscrepancyBound}
For a function $f \in L^{2}([0,1)^{d})$ and a fixed shift $\bs \in [0,1)^{d}$, at every step,
\[
\mathbb{E}_{\nu_{\bs}}[\mathsf{disc}_{t}(P_{<h}f_{\bs})^{2}] \leq O_{d}(\log^{d}n)\lVert P_{<h}f_{\bs}\rVert_{\mathsf H}^{2},
\]
where $O_d(\cdot)$ hides factors depending only on $d$.
\end{lemma}

\begin{proof}[Proof of \Cref{lem:low-frequencyDiscrepancyBound}.]
    Using the Haar decomposition of $P_{<h}f_{\bs}$, we write $\mathsf{disc}_{t}(P_{<h}f_{\bs})$ as:
    \begin{align}\label{eqn:low-frequencyDiscrepancyBoundEqn1}
        \mathsf{disc}_{t}(P_{<h}f_{\bs}) = \sum_{z \in A_{t}}x_{t}(\bz)P_{<h}f_{\bs}(\bz) = \sum_{z \in A_{t}}x_{t}(\bz)\sum_{\substack{S \subseteq [d], \\ S \neq \emptyset}}\sum_{\substack{\mathbf{j},\boldsymbol{\ell} \in P_{<h}^{S}}}\langle \Psi_{S,\mathbf{j},\boldsymbol{\ell}}, P_{<h}f_{\bs}\rangle \,\Psi_{S,\mathbf{j},\boldsymbol{\ell}}(\bz).
    \end{align}
   Rearranging sums, we can express the left-hand side as a linear combination of \emph{weighted-discrepancies} with respect to the Haar functions:
    \begin{align}\label{eqn:low-frequencyDiscrepancyBoundEqn2}
        \sum_{z \in A_{t}}x_{t}(\bz)\sum_{\substack{S \subseteq [d], \\ S \neq \emptyset}}\sum_{\substack{\mathbf{j},\boldsymbol{\ell} \in P_{<h}^{S}}}\langle \Psi_{S,\mathbf{j},\boldsymbol{\ell}}, P_{<h}f_{\bs}\rangle \,\Psi_{S,\mathbf{j},\boldsymbol{\ell}}(\bz) = \sum_{\substack{S \subseteq [d], \\ S \neq \emptyset}}\sum_{\substack{\mathbf{j},\boldsymbol{\ell} \in P_{<h}^{S}}}\langle \Psi_{S,\mathbf{j},\boldsymbol{\ell}}, P_{<h}f_{\bs}\rangle\sum_{z \in A_{t}}x_{t}(\bz) \,\Psi_{S,\mathbf{j},\boldsymbol{\ell}}(\bz).
    \end{align}
    Recall from \eqref{eqn:familiaritywithHaareqn3} that for a nonempty $S \subseteq [d]$, the Haar function $\Psi_{S,\mathbf{j},\boldsymbol{\ell}}$ can be expressed as a linear combination of signed boxes. For a fixed subset $S\subseteq[d]$ and corresponding multi-indices $\mathbf{j},\boldsymbol{\ell}$ indexed by $S$, using \eqref{eqn:familiaritywithHaareqn3} we can express the discrepancy with respect to $\Psi_{S,\mathbf{j},\boldsymbol{\ell}}$ as a signed sum of discrepancies over dyadic boxes.
\begin{align}
\sum_{z \in A_t} x_t(\bz)\Psi_{S,\mathbf{j},\boldsymbol{\ell}}(\bz)
&=
\sum_{z \in A_t} x_t(\bz)
2^{|\mathbf{j}|/2}
\sum_{T\subseteq S}
(-1)^{|S\setminus T|}
B_{S,\mathbf{j}+\mathbf{1},
2\boldsymbol{\ell}+\mathbf{1}_{S\setminus T}}(\bz)
\notag\\
&=
2^{|\mathbf{j}|/2}
\sum_{T\subseteq S}
(-1)^{|S\setminus T|}
\sum_{z\in A_t}
x_t(\bz)
B_{S,\mathbf{j}+\mathbf{1},
2\boldsymbol{\ell}+\mathbf{1}_{S\setminus T}}(\bz).
\label{eqn:low-frequencyDiscrepancyBoundEqn5}
\end{align}

Substituting \eqref{eqn:low-frequencyDiscrepancyBoundEqn5} into \eqref{eqn:low-frequencyDiscrepancyBoundEqn2}, and then into \eqref{eqn:low-frequencyDiscrepancyBoundEqn1}, gives:
\[
\mathsf{disc}_{t}(P_{<h}f_{\bs}) = \sum_{\substack{S \subseteq [d], \\ S \neq \emptyset}}\sum_{\substack{\mathbf{j},\boldsymbol{\ell} \in P_{<h}^{S}}}2^{|\mathbf{j}|/2}\langle \Psi_{S,\mathbf{j},\boldsymbol{\ell}},P_{<h}f_{\bs}\rangle \sum_{T \subseteq S}(-1)^{|S \setminus T|}\sum_{z \in A_{t}}x_{t}(\bz)B_{S,\mathbf{j} + \mathbf{1},2\boldsymbol{\ell} + \mathbf{1}_{S \setminus T}}(\bz).
\]

This helps us express $\mathsf{disc}_{t}(P_{<h}f_{\bs})$ as the inner product between dyadic box discrepancies and suitably scaled Haar coefficients. For each nonempty $S\subseteq[d]$ and each pair $\mathbf{j},\boldsymbol{\ell}$, there are $2^{|S|}$ dyadic-box discrepancies, each weighted by $2^{|\mathbf{j}|/2}\langle P_{<h}f_{\bs},\Psi_{S,\mathbf{j},\boldsymbol{\ell}}\rangle$. Thus, the second moment of $\mathsf{disc}_t(P_{<h}f_{\bs})$ can be bounded as follows\footnote{This bound follows from the standard moment estimate for a sub-Gaussian random variable $X$:
$(\mathbb{E}[|X|^p])^{1/p}\leq O(\sqrt{p})\lVert X\rVert_{\psi_2}$, where $\lVert X\rVert_{\psi_2}$ denotes its sub-Gaussian norm (see \cite{Ver18book}).}:
\begin{align*}
\mathbb{E}_{\nu_{\bs}}[\mathsf{disc}_{t}(P_{<h}f_{\bs})^{2}] &\leq O(\log^{d}n)(\sum_{\substack{S \subseteq [d], \\ S \neq \emptyset}}2^{|S|}\sum_{\substack{\mathbf{j},\boldsymbol{\ell} \in P_{<h}^{S}}}2^{|\mathbf{j}|}\langle \Psi_{S,\mathbf{j},\boldsymbol{\ell}},P_{<h}f_{\bs}\rangle^{2})\\
&\leq O(2^{d}\log^{d}n)\lVert P_{<h}f_{\bs}\rVert_{\mathsf H}^{2} = O_{d}(\log^{d}n)\lVert P_{<h}f_{\bs}\rVert_{\mathsf H}^{2},
\end{align*}
where the second inequality follows by the definition of the Haar--Besov seminorm of $P_{<h}f_{\bs}$ (see \Cref{defn:HaarBesovSeminorm}) and $2^{|S|} \leq 2^{d}$ for all $S \subseteq [d]$. This completes the argument.
\end{proof}

We are now ready to prove \Cref{lem:lowfrequencyError}.

\begin{proof}
    To bound the second moment of $\mathsf{err}(A_T,P_{<h}f_{\bs})$, we apply \Cref{prop:SecondMomentBoundsOnError} and bound each resulting term separately. 
    \begin{align}\label{eqn:lowFrequencyErrorEqn1}
      \mathbb{E}_{\nu_{\bs}}[\mathsf{err}(A_{T},P_{<h}f_{\bs})^{2}] = \mathbb{E}_{\nu_{\bs}}[\mathsf{err}(A_{0},P_{<h}f_{\bs})^{2}] + \sum_{t < T} \frac{\mathbb{E}_{\nu_{\bs}}[\mathsf{disc}_{t}(P_{<h}f_{\bs})^{2}]}{|A_{t}|^{2}}.
    \end{align}
    By \Cref{lem:low-frequencyDiscrepancyBound}, the second term can be bounded by,
   \begin{align}\label{eqn:lowFrequencyErrorEqn2}
        \sum_{t < T} \frac{\mathbb{E}_{\nu_{\bs}}[\mathsf{disc}_{t}(P_{<h}f_{\bs})^{2}]}{|A_{t}|^{2}} \leq O_{d}(\log^{d}n)\frac{\lVert P_{<h}f_{\bs}\rVert_{\mathsf H}^{2}}{n^{2}}(\sum_{j \geq 0} 2^{-j}) = O_{d}(\log^{d}n)\frac{\lVert P_{<h}f_{\bs}\rVert_{\mathsf H}^{2}}{n^{2}}.
   \end{align}
   As $A_{0}$ is a collection of $n^{2}$ i.i.d. samples from $[0,1)^{d}$, the first term is the scaled standard deviation.
\begin{align}
\mathbb{E}_{\nu_{\bs}}[\mathsf{err}(A_{0}, P_{< h}f_{\bs})^{2}] &=
\frac{1}{|A_{0}|^{2}}\mathbb{E}_{\nu_{\bs}}[\bigl(\sum_{z \in A_{0}}P_{< h}f_{\bs}(z)\bigr)^{2}]\notag \\
&=
\frac{1}{|A_{0}|^{2}}\sum_{z \in A_{0}}\lVert P_{< h}f_{\bs}\rVert_{2}^{2} 
=\frac{\lVert P_{< h}f_{\bs}\rVert_{2}^{2}}{n^{2}} \leq \frac{\lVert P_{< h}f_{\bs}\rVert_{\mathsf H}^{2}}{n^{2}}, \label{eqn:lowFrequencyErrorEqn3}
\end{align}
where the first equality uses that $f_{\bs}$ has mean zero, and the second uses that the samples are independent and mean zero. Plugging \eqref{eqn:lowFrequencyErrorEqn2} and \eqref{eqn:lowFrequencyErrorEqn3} in \eqref{eqn:lowFrequencyErrorEqn1} completes the proof.
\end{proof}

\subsubsection{Error on the High-Frequency Components is Negligible}
\label{subsubsec:highfrequencyanalysis}
In this section, we show that   the error incurred 
by Algorithm \ref{subsubsecn:algDescription}
on the high-frequency Haar projection of $f$ is negligible. 
Recall from \Cref{subsubsec:finalProof} that $P_{\geq h}$ denotes the orthogonal projection onto the subspace of $L^2([0,1)^d)$ spanned by the Haar functions $\Psi_{S,\mathbf{j},\boldsymbol{\ell}}$ over all nonempty sets $S\subseteq[d]$ satisfying $\max_{i\in S}\mathbf{j}_i\geq h$ and $0\leq\boldsymbol{\ell}_i<2^{\mathbf{j}_i}$ for every $i\in S$.

\begin{lemma}\label{lem:highFrequency}
For a function $f \in L^{2}([0,1)^{d})$ and a fixed shift $\bs \in [0,1)^{d}$, the second moment of the high-frequency error of Algorithm \ref{subsubsecn:algDescription} is at most:
    \[
    \mathbb{E}_{\nu_{\bs}}[\mathsf{err}(A_{T}, P_{\geq h}f_{\bs})^{2} ] \leq \frac{O_d(\log^{d}n)}{\mathsf{poly}(n)}\lVert P_{\geq h}f_{\bs}\rVert_{\mathsf H}^{2},
    \]
    where $\nu_{\bs}$ denotes the distribution of the output point set conditioned on the shift $\bs$. 
\end{lemma}

For the high-frequency component of $f_{\bs}$, the Haar--Besov seminorm dominates the $L^2$ norm by a factor of $\mathsf{poly}(n): \|P_{\geq h}f_{\bs}\|_2 \leq \|P_{\geq h}f_{\bs}\|_\mathsf{H}/\poly(n)$. For the high-frequency component, we use the sub-Gaussianity of the signing $x_t$ to bound the error in terms of $\lVert P_{\geq h}f_{\bs}\rVert_2$, showing that this contribution is negligible.

\begin{proof}[Proof of \Cref{lem:highFrequency}.]
Using \Cref{prop:SecondMomentBoundsOnError}, we can express the second moment of the high-frequency component as:
\begin{equation}\label{eqn:highFrequencyErrorEqn1}
          \mathbb{E}_{\nu_{\bs}}[\mathsf{err}(A_{T},P_{\geq h}f_{\bs})^{2}] = \mathbb{E}_{\nu_{\bs}}[\mathsf{err}(A_{0},P_{\geq h}f_{\bs})^{2}] + \sum_{t < T} \frac{\mathbb{E}_{\nu_{\bs}}[\mathsf{disc}_{t}(P_{\geq h}f_{\bs})^{2}]}{|A_{t}|^{2}}.
\end{equation}

By \eqref{eqn:lowFrequencyErrorEqn3}, the first term is bounded by, 
\begin{align}\label{eqn:highFrequencyErrorEqn1.5}
    \mathbb{E}_{\nu_{\bs}}[\mathsf{err}(A_{0},P_{\geq h}f_{\bs})^{2}] \leq \frac{\lVert P_{\geq h}f_{\bs}\rVert_{2}^{2}}{\mathsf{poly}(n)}.
\end{align}

To bound the second term, we use the sub-Gaussianity of $x_t$ to control the second moment of $\mathsf{disc}_t(P_{\geq h}f_{\bs})$ at each step $t$.
    \[
        \mathbb{E}_{\nu_{\bs}}[\mathsf{disc}_{t}(P_{\geq h}f_{\bs})^{2}] = \mathbb{E}_{\nu_{\bs}}[\,\bigl (\sum_{z \in A_{t}} x_{t}(\bz)P_{\geq h}f_{\bs}(\bz)\bigr)^{2}] \leq O(\log^{d}n)|A_{t}| \lVert P_{\geq h}f_{\bs}\rVert_{2}^{2},
\]
where we use that for each point $z \in A_{t}$, its marginal is still uniform. Summing across all steps,
\begin{equation}\label{eqn:highFreqFourthEqn}
    \sum_{t < T}\frac{\mathbb{E}_{\nu_{\bs}}[\mathsf{disc}_{t}(P_{\geq h}f_{\bs})^{2}]}{|A_{t}|^{2}} \leq O(\log^{d}n)\frac{\lVert P_{\geq h}f_{\bs}\rVert_{2}^{2}}{|A_{T}|}\bigl(\sum_{j \geq 0}2^{-j}\bigr) = O(\log^{d}n)\frac{\lVert P_{\geq h}f_{\bs}\rVert_{2}^{2}}{|A_{T}|}.
\end{equation}

As the Haar--Besov seminorm scales with frequencies, for frequencies above $\Omega(\log n)$, it dominates the $L^{2}$-norm by a factor of $\mathsf{poly}(n)$.
\begin{align}
\lVert P_{\geq h}f_{\bs}\rVert_{2}^{2} &= \sum_{\substack{S \subseteq [d], \\ S \neq \emptyset}}\sum_{\substack{\mathbf{j},\boldsymbol{\ell} \in P_{\geq h}^{S}}}\langle \Psi_{S,\mathbf{j},\boldsymbol{\ell}},f_{\bs}\rangle^{2}\notag \\
&\leq \frac{1}{\mathsf{poly}(n)}\sum_{\substack{S \subseteq [d], \\ S \neq \emptyset}}\sum_{\substack{\mathbf{j},\boldsymbol{\ell} \in P_{\geq h}^{S}}}2^{|\mathbf{j}|}\langle f_{\bs}, \Psi_{S, \mathbf{j},\boldsymbol{\ell}}\rangle^{2} = \frac{\lVert P_{\geq h}f_{\bs}\rVert_{\mathsf H}^{2}}{\mathsf{poly}(n)} \label{eqn:highFrequencyErrorFifthEqn},
    \end{align}
where we use that the multi-index vector $\mathbf{j}$ has at least one entry is $\Omega(\log n)$. Combining \eqref{eqn:highFrequencyErrorFifthEqn}, \eqref{eqn:highFreqFourthEqn}, \eqref{eqn:highFrequencyErrorEqn1.5}, and \eqref{eqn:highFrequencyErrorEqn1}, we complete the argument.
\end{proof}

\bigskip

\subsection{Equivalence Between Haar--Besov Seminorm and \texorpdfstring{$\sigma_{\mathsf{SO}}$}{sigmaSO}}\label{subsec:equivalence}

In this section, we characterize the smoothed-out variation $\sigma_{\mathsf{SO}}(f)$, introduced by Bansal and Jiang \cite{BJ25a} to bound their numerical integration error, in terms of the Haar--Besov seminorm of $f$ averaged over shifts $\bs\in[0,1)^d$. 

\shiftAveragedHaarEquivalence*

This follows from standard results in functional analysis \cite{vybiral2006function,triebel2019function}, but we include a full proof for completeness. The proof proceeds in three main steps.

\smallskip
\noindent\textbf{Expansion in the Fourier Basis.} We begin by expressing the Haar--Besov seminorm of $f_{\bs}$, averaged over shifts $\bs\in[0,1)^d$, in terms of the Haar--Besov seminorms of the Fourier characters $e_{\mathbf{k}}$ for $\mathbf{k}\in\mathbb{Z}^d$. We refer the reader to \Cref{sec:prelims} for notation on Fourier analysis over $[0,1)^d$.
\begin{proposition}\label{prop:HaarBesovofftoHaarBesovOfExp}
    For a function $f \in L^{2}([0,1)^{d})$, its Haar--Besov seminorm averaged over random shifts $\bs \sim \mathrm{Unif}([0,1)^{d})$ can be expressed as:
    \[
    \mathbb{E}_{\bs \sim [0,1)^{d}}[\lVert f_{\bs}\rVert_{\mathsf H}^{2}] = \sum_{\mathbf{k} \in \mathbb{Z}^{d}}|\widehat{f}_{\mathbf{k}}|^{2}\sum_{\substack{S\subseteq[d]\\S\neq\emptyset}}
    \sum_{\mathbf{j},\boldsymbol{\ell}}
    2^{|\mathbf{j}|}|\langle e_{\mathbf{k}},\Psi_{S,\mathbf{j},\boldsymbol{\ell}}\rangle|^{2}
    \]
\end{proposition}

\begin{proof}[Proof of \Cref{prop:HaarBesovofftoHaarBesovOfExp}.]
The shift-averaged squared Haar--Besov seminorm of $f$ can be written as a weighted sum of its shift-averaged squared Haar coefficients:
    \[
        \mathbb{E}_{\bs \sim [0,1)^{d}}[\lVert f_{\bs}\rVert_{\mathsf H}^2]
    =
    \sum_{\substack{S\subseteq[d]\\S\neq\emptyset}}
    \sum_{\mathbf{j},\boldsymbol{\ell}}
    2^{|\mathbf{j}|}
    \mathbb{E}_{\bs \sim [0,1)^{d}}[|\langle f_{\bs},\Psi_{S,\mathbf{j},\boldsymbol{\ell}}\rangle|^2].
    \]
Observe that the shift average of each squared Haar coefficient is exactly the squared $L^2$-norm of the function $g: \bs \mapsto \langle f_{\bs},\Psi_{S,\mathbf{j},\boldsymbol{\ell}}\rangle$ on $[0,1)^d$. By Parseval's identity, this $L^2$-norm can be expressed in terms of its Fourier coefficients (see \Cref{sec:prelims}). To do so, we compute its Fourier coefficient corresponding to $e_{\mathbf{k}}$ for each $\mathbf{k}\in\mathbb{Z}^d$.
    \begin{align*}\label{prop:HaarBesovofftohaarBesovofExpEqn2}
        \widehat{g}_{\mathbf{k}}
    &=
    \int_{\bs \in [0,1)^{d}} g(\bs)e_{\mathbf{k}}(-\bs) d\bs  \\
    &= \int_{\substack{\bz\in[0,1)^d\\ \bs\in[0,1)^d}} f(\bs + \bz)\Psi_{S, \mathbf{j},\boldsymbol{\ell}}(\bz)e_{\mathbf{k}}(-\bs) d\bz d\bs \\
    &= \widehat{f}_{\mathbf{k}}\int_{\bz \in [0,1)^{d}}e_{\mathbf{k}}(\bz)\Psi_{S,\mathbf{j},\boldsymbol{\ell}} d\bz = \widehat{f}_{\mathbf{k}}\,\langle e_{\mathbf{k}},\Psi_{S,\mathbf{j},\boldsymbol{\ell}}\rangle d\bz.
    \end{align*}

Expressing the $L^{2}$-norm of $g$ using its Fourier coefficients completes the proof.
    \begin{align*}
        \sum_{\substack{S\subseteq[d]\\S\neq\emptyset}}
    \sum_{\mathbf{j},\boldsymbol{\ell}}
    2^{|\mathbf{j}|}
    \mathbb{E}_{\bs \sim [0,1)^{d}}[|\langle f_{\bs},\Psi_{S,\mathbf{j},\boldsymbol{\ell}}\rangle|^2] &=   \sum_{\substack{S\subseteq[d]\\S\neq\emptyset}}
    \sum_{\mathbf{j},\boldsymbol{\ell}}
    2^{|\mathbf{j}|}\sum_{\mathbf{k} \in \mathbb{Z}^{d}}|\widehat{f}_{\mathbf{k}}|^{2}|\langle e_{\mathbf{k}},\Psi_{S,\mathbf{j},\boldsymbol{\ell}}\rangle|^{2} \\
    &= \sum_{\mathbf{k} \in \mathbb{Z}^{d}}|\widehat{f}_{\mathbf{k}}|^{2}\sum_{\substack{S\subseteq[d]\\S\neq\emptyset}} \sum_{\mathbf{j},\boldsymbol{\ell}}
    2^{|\mathbf{j}|}|\langle e_{\mathbf{k}},\Psi_{S,\mathbf{j},\boldsymbol{\ell}}\rangle|^{2}. \qedhere
    \end{align*}
\end{proof}
We next show that $e_{\mathbf{k}}$ has nonzero Haar coefficients only when $S=\operatorname{supp}(\mathbf{k})$, where $\operatorname{supp}(\mathbf{k})=\{i\in[d]:\mathbf{k}_i\neq 0\}$. We write $S_{\mathbf{k}}$ as shorthand for $\operatorname{supp}(\mathbf{k})$.

\begin{corollary}\label{corr:HaarBesovofftoHaarBesovOfExp}
       For a function $f \in L^{2}([0,1)^{d})$, its squared Haar--Besov seminorm averaged over random shifts $\bs \sim \mathrm{Unif}([0,1)^{d})$ can be expressed as:
    \[
    \mathbb{E}_{\bs \sim [0,1)^{d}}[\lVert f_{\bs}\rVert_{\mathsf H}^{2}] = \sum_{\substack{\mathbf{k}\in\mathbb{Z}^d,\\ \mathbf{k}\neq\mathbf{0}}}|\widehat{f}_{\mathbf{k}}|^{2}
    \sum_{\mathbf{j},\boldsymbol{\ell}}
    2^{|\mathbf{j}|}|\langle e_{\mathbf{k}},\Psi_{S_{\mathbf{k}},\mathbf{j},\boldsymbol{\ell}}\rangle|^{2}.
    \]
\end{corollary}

\begin{proof}
    Consider the expansion of the shift-averaged Haar--Besov seminorm in Proposition \ref{prop:HaarBesovofftoHaarBesovOfExp}. For a subset $S$, its inner product with $e_{\mathbf{k}}$ can be written as a product over coordinates in $[d]$.
    \[
    \langle e_{\mathbf{k}}, \Psi_{S,\mathbf{j},\boldsymbol{\ell}}\rangle  = \displaystyle\prod_{i \in S}\langle e_{\mathbf{k}_{i}},\psi_{\mathbf{j}_{i},\boldsymbol{\ell}_{i}}\rangle \displaystyle\prod_{i \in [d]\setminus S}\langle e_{\mathbf{k}_{i}}, B_{0,0}\rangle,
    \]
    where $B_{0,0}$ denotes the constant $1$ function over $[0,1)$ (see \Cref{sec:prelims}).
For $i\in S$, if $\mathbf{k}_i=0$, then $e_{\mathbf{k}_i}=1$, which is orthogonal to the corresponding Haar function. Hence, the inner product becomes $0$. Similarly, for $i\in[d]\setminus S$, if $\mathbf{k}_i\neq 0$, then $e_{\mathbf{k}_i}$ is orthogonal to the constant function on $[0,1)$, and the inner product again becomes $0$. Therefore, $e_{\mathbf{k}}$ for $\mathbf{k} \neq 0$, only $S_{\mathbf{k}}$ contributes. Since we sum only over nonempty subsets $S$, the constant character $e_{\mathbf{0}}$ does not contribute.
\end{proof}

\smallskip
\noindent\textbf{From Multi-Dimension to One-Dimension.} Next, we express the projection of $e_{\mathbf{k}}$ onto $\Psi_{S_{\mathbf{k}},\mathbf{j},\boldsymbol{\ell}}$ as the product of the projections of the one-dimensional characters $e_{\mathbf{k}_i}$ onto $\psi_{\mathbf{j}_i,\boldsymbol{\ell}_i}$. This reduces the multidimensional calculation to a collection of one-dimensional calculations.

\begin{proposition}\label{prop:HaarBesovSOMultiToSingle}
For any nonzero vector $\mathbf{k}\in\mathbb{Z}^d$, the squared Haar--Besov seminorm of the Fourier character $e_{\mathbf{k}}$ can be written as follows:
\[
\sum_{\mathbf{j},\boldsymbol{\ell}}2^{|\mathbf{j}|}|\langle e_{\mathbf{k}},\Psi_{S_{\mathbf{k}},\mathbf{j},\boldsymbol{\ell}}\rangle|^{2} = \displaystyle\prod_{i \in S_{\mathbf{k}}}\bigl(\sum_{j \geq 0}2^{j}\sum_{\ell = 0}^{2^{j}-1}|\langle e_{\mathbf{k}_{i}},\psi_{j,\ell}\rangle|^{2}\bigr).
\]
    
\end{proposition}

\begin{proof}[Proof of \Cref{prop:HaarBesovSOMultiToSingle}]
    The proof follows from a straightforward \emph{tensorization} of the inner product between $e_{\mathbf{k}}$ and $\Psi_{S_{\mathbf{k}},\mathbf{j},\boldsymbol{\ell}}$. In particular,
    \begin{align*}
    \sum_{\mathbf{j},\boldsymbol{\ell}}
    2^{|\mathbf{j}|}|\langle e_{\mathbf{k}},\Psi_{S_{\mathbf{k}},\mathbf{j},\boldsymbol{\ell}}\rangle|^{2} &=  \sum_{\mathbf{j},\boldsymbol{\ell}}
    2^{|\mathbf{j}|}\displaystyle\prod_{i \in S_{\mathbf{k}}}|\langle e_{\mathbf{k}_{i}},\psi_{\mathbf{j}_{i},\boldsymbol{\ell}_{i}}\rangle|^{2} 
    = \sum_{\mathbf{j}}
    2^{|\mathbf{j}|}\sum_{\boldsymbol{\ell}}\displaystyle\prod_{i \in S_{\mathbf{k}}}|\langle e_{\mathbf{k}_{i}},\psi_{\mathbf{j}_{i},\boldsymbol{\ell}_{i}}\rangle|^{2}\\
     &= \sum_{\mathbf{j}}\displaystyle\prod_{i \in S_{\mathbf{k}}}2^{\mathbf{j}_{i}}(\sum_{\ell = 0}^{2^{\mathbf{j}_{i}}-1}|\langle e_{\mathbf{k}_{i}},\psi_{\mathbf{j}_{i},\ell}\rangle|^{2}) = \displaystyle\prod_{i \in S_{\mathbf{k}}}\bigl(\sum_{j \geq 0}2^{j}\sum_{\ell = 0}^{2^{j}-1}|\langle e_{\mathbf{k}_{i}},\psi_{j,\ell}\rangle|^{2}\bigr).
    \end{align*}
This tensorization follows because multidimensional dyadic boxes are tensor products of one-dimensional dyadic intervals.
\end{proof}

\smallskip
\noindent\textbf{Analysis in One-Dimension.} It remains to show that, for every nonzero $k\in\mathbb{Z}$, the squared Haar--Besov seminorm of $e_k(z) :=\exp(2\pi i kz)$ is of order $|k|$. In particular, 
\begin{proposition}\label{prop:HaarBesovSemiNormofSingleDimExp}
    For every $k \neq 0$, the function $e_{k}(z)$ has squared Haar--Besov seminorm
    \[
    \|e_k\|_\mathsf{H}^2 = \sum_{j \geq 0}2^{j}\sum_{\ell = 0}^{2^{j}-1}|\langle e_{k}, \psi_{j,\ell}\rangle|^{2} = \Theta (|k|).
    \]
\end{proposition}

\begin{proof}
We prove this by estimating $\langle e_k,\psi_{j,\ell}\rangle$ and summing the resulting bounds over $j$ and $\ell$. Since each $\psi_{j,\ell}$ is a scaled and shifted copy of $\psi_{0,0}$, we express $\langle e_k,\psi_{j,\ell}\rangle$ using the corresponding coefficient for $\psi_{0,0}$ for a modified exponential $\exp(2\pi \lambda z i)$. We make this precise below. To this end, define $g(\lambda)=\int_0^1\exp(2\pi i\lambda z)\psi_{0,0}(z)\,dz$. Then, the $(j,\ell)^{\text{th}}$-Haar coefficient of $e_{k}$ can be expressed as:
\[
\langle e_{k}, \psi_{j,\ell}\rangle = \int_{\ell 2^{-j}}^{(\ell+1)2^{-j}}\exp(2\pi i k z)\psi_{j,\ell}(z)dz = 2^{j/2}\int_{\ell 2^{-j}}^{(\ell+1)2^{-j}}\exp(2\pi i k z)\psi_{0,0}(2^{j}z - \ell)dz,
\]
where we used that $\psi_{j,\ell}(z) = 2^{j/2}\psi_{0,0}(2^{j}z - \ell)$. Applying a change of variable, we get
\begin{align}
2^{j/2}\int_{\ell 2^{-j}}^{(\ell+1)2^{-j}}\exp(2\pi i k z)\psi_{0,0}(2^{j}z - \ell)dz &= 2^{-j/2}\int_{0}^{1}\exp(2\pi i k2^{-j}(u+\ell)\psi_{0,0}(u)du\notag
\\ &=2^{-j/2}\exp(2\pi i k\ell2^{-j})g(k2^{-j})\label{prop:HaarBesovSemiNormofSingleDimExpEqn3}.
\end{align}
A standard calculation \cite{Dau92} gives $|g(\lambda)|=2\sin^2(\pi\lambda/2)/(\pi|\lambda|)$. Hence, $|g(\lambda)|\asymp|\lambda|$ for $0<|\lambda|\leq 1$, while $|g(\lambda)|\lesssim|\lambda|^{-1}$ for $|\lambda|\geq 1$. We split the sum at $\lfloor\log_2|k|\rfloor$ and substitute \eqref{prop:HaarBesovSemiNormofSingleDimExpEqn3} into the left-hand side of \Cref{prop:HaarBesovSemiNormofSingleDimExp}. Applying the two estimates over the corresponding ranges gives:
\begin{align*}
    \sum_{j \geq 0}2^{j}\sum_{\ell = 0}^{2^{j}-1}|\langle e_{k}, \psi_{j,\ell}\rangle|^{2} &= \sum_{j \geq 0}2^{j}|g(k2^{-j})|^{2} \lesssim \sum_{j = 0}^{\lfloor\log_2|k|\rfloor}|k|^{-2}2^{3j} + \sum_{j > \lfloor\log_2|k|\rfloor}|k|^{2}2^{-j} \lesssim |k|,
\end{align*}
The first geometric sum is $O(|k|^3)$, so the factor $|k|^{-2}$ reduces its contribution to $O(|k|)$. Similarly, the second geometric sum is $O(|k|^{-1})$, and multiplying it by $|k|^2$ again gives $O(|k|)$. For the lower bound, we retain only the terms with $j>\lfloor\log_2|k|\rfloor$. Their contribution is $\Omega(|k|)$.
\end{proof}

We are ready to prove \Cref{thm:shift-averaged-mixed-haar-equivalence}. 
\begin{proof}[Proof of \Cref{thm:shift-averaged-mixed-haar-equivalence}.]
   Combining \Cref{corr:HaarBesovofftoHaarBesovOfExp} with \Cref{prop:HaarBesovSOMultiToSingle} and applying the estimates from \Cref{prop:HaarBesovSemiNormofSingleDimExp}, we obtain:
\begin{align*}
\mathbb{E}_{\bs \sim [0,1)^{d}}[\lVert f_{\bs}\rVert_{\mathsf H}^{2}] &= \sum_{\substack{\mathbf{k}\in\mathbb{Z}^d,\\ \mathbf{k}\neq\mathbf{0}}}|\widehat{f}_{\mathbf{k}}|^{2}
    \sum_{\mathbf{j},\boldsymbol{\ell}}
    2^{|\mathbf{j}|}|\langle e_{\mathbf{k}},\Psi_{S_{\mathbf{k}},\mathbf{j},\boldsymbol{\ell}}\rangle|^{2}\\
    &= \sum_{\substack{\mathbf{k}\in\mathbb{Z}^d,\\ \mathbf{k}\neq\mathbf{0}}}|\widehat{f}_{\mathbf{k}}|^{2}\displaystyle\prod_{i \in S_{\mathbf{k}}}\bigl(\sum_{j \geq 0}2^{j}\sum_{\ell = 0}^{2^{j}-1}|\langle e_{\mathbf{k}_{i}},\psi_{j,\ell}\rangle|^{2}\bigr) 
    = \Theta_{d}(1)\sum_{\substack{\mathbf{k}\in\mathbb{Z}^d,\\ \mathbf{k}\neq\mathbf{0}}}|\widehat{f}_{\mathbf{k}}|^{2}\displaystyle\prod_{i \in S_{\mathbf{k}}}\max(1, |\mathbf{k}_{i}|),
\end{align*}
where the first equality follows by \Cref{prop:HaarBesovSOMultiToSingle} and the second equality follows by \Cref{prop:HaarBesovSemiNormofSingleDimExp}. The factor $\Theta_d(1)$ arises because a constant factor is accumulated in each dimension, resulting in a factor of $\exp(O(d))$, which is constant for fixed $d$. This completes the argument.
\end{proof}

\bibliographystyle{alphaurl}
\bibliography{bib.bib}

@article{AA26,
  title={Optimal Online Discrepancy Minimization in Linear Time},
  author={Aden-Ali, Ishaq},
  journal={arXiv preprint arXiv:2607.04388},
  year={2026}
}

@book{Ver18book,
  title={High-dimensional probability: An introduction with applications in data science},
  author={Vershynin, Roman},
  volume={47},
  year={2018},
  publisher={Cambridge university press}
}

@inproceedings{BJ25a,
  title={{Quasi-Monte Carlo Beyond Hardy-Krause}},
  author={Bansal, Nikhil and Jiang, Haotian},
  booktitle={ Symposium on Discrete Algorithms (SODA)},
  pages={2051--2075},
  year={2025},
  organization={SIAM}
}

@inproceedings{ALS21,
    AUTHOR = {Alweiss, Ryan and Liu, Yang P. and Sawhney, Mehtaab},
     TITLE = {Discrepancy minimization via a self-balancing walk},
 BOOKTITLE = {
              Symposium on Theory of Computing, {STOC}},
     PAGES = {14--20},
      YEAR = {2021},
      ISBN = {978-1-4503-8053-9},
       DOI = {10.1145/3406325.3450994},
       URL = {https://doi.org/10.1145/3406325.3450994}
}

@book{vybiral2006function,
  title={Function spaces with dominating mixed smoothness},
  author={Vybíral, Jan},
  series={Dissertationes Mathematicae},
  volume={436},
  pages={1--73},
  year={2006},
  doi={10.4064/dm436-0-1}
}

@book{triebel2019function,
  title={Function Spaces with Dominating Mixed Smoothness},
  author={Triebel, Hans},
  series={EMS Series of Lectures in Mathematics},
  publisher={European Mathematical Society},
  year={2019},
  doi={10.4171/195}
}

@article{Kok42,
  author  = {Koksma, Jurjen Ferdinand},
  title   = {A General Theorem from the Theory of Uniform Distribution Modulo 1},
  journal = {Mathematica, Zutphen. B},
  volume  = {11},
  pages   = {7--11},
  year    = {1942}
}

@article{Hla61,
  author  = {Hlawka, Edmund},
  title   = {Funktionen von beschr{\"a}nkter Variation in der Theorie der Gleichverteilung},
  journal = {Annali di Matematica Pura ed Applicata},
  volume  = {54},
  number  = {1},
  pages   = {325--333},
  year    = {1961}
}

@article{Zar68,
  author  = {Zaremba, S. C.},
  title   = {Some Applications of Multidimensional Integration by Parts},
  journal = {Annales Polonici Mathematici},
  volume  = {21},
  number  = {1},
  pages   = {85--96},
  year    = {1968}
}

@misc{Owe13,
  author = {Owen, Art B.},
  title  = {{Monte Carlo} Theory, Methods and Examples},
  year   = {2013},
  url    = {https://artowen.su.domains/mc/}
}

@book{Mat09,
  author    = {Matou{\v{s}}ek, Ji{\v{r}}{\'\i}},
  title     = {Geometric Discrepancy: An Illustrated Guide},
  volume    = {18},
  publisher = {Springer Science \& Business Media},
  year      = {2009}
}

@misc{WikiNumIntegration,
  author       = {{Wikipedia contributors}},
  title        = {Numerical integration},
  year         = {2026},
  howpublished = {\url{https://en.wikipedia.org/w/index.php?title=Numerical_integration&oldid=1340541135}},
  note         = {Wikipedia, The Free Encyclopedia. Accessed July 24, 2026}
}

@article{DS2000,
  author  = {Dongarra, Jack and Sullivan, Francis},
  title   = {Guest Editors' Introduction to the Top 10 Algorithms},
  journal = {Computing in Science \& Engineering},
  year    = {2000},
  volume  = {2},
  number  = {1},
  pages   = {22--23},
  doi     = {10.1109/MCISE.2000.814652}
}

@inproceedings{Fishman05,
  title={{A First Course in Monte Carlo}},
  author={George S. Fishman},
  year={2005},
  url={https://api.semanticscholar.org/CorpusID:57302552}
}

@book{Glasserman03,
  author    = {Glasserman, Paul},
  title     = {Monte Carlo Methods in Financial Engineering},
  series    = {Stochastic Modelling and Applied Probability},
  volume    = {53},
  publisher = {Springer},
  address   = {New York, NY},
  year      = {2003},
  doi       = {10.1007/978-0-387-21617-1}
}

@book{DKPS13,
  author    = {Dick, Josef and Kuo, Frances Y. and Peters, Gareth W. and Sloan, Ian H.},
  title     = {{Monte Carlo} and {Quasi-Monte Carlo} Methods},
  publisher = {Springer},
  address   = {Berlin, Heidelberg},
  year      = {2013}
}

@book{Lem09,
  author    = {Lemieux, Christiane},
  title     = {{Monte Carlo} and {Quasi-Monte Carlo} Sampling},
  publisher = {Springer},
  address   = {New York},
  year      = {2009}
}

@book{DP10,
  author    = {Dick, Josef and Pillichshammer, Friedrich},
  title     = {Digital Nets and Sequences: Discrepancy Theory and {Quasi-Monte Carlo} Integration},
  publisher = {Cambridge University Press},
  year      = {2010}
}

@book{Nie92,
  author    = {Niederreiter, Harald},
  title     = {Random Number Generation and {Quasi-Monte Carlo} Methods},
  publisher = {SIAM},
  year      = {1992}
}

@book{Dau92,
  author    = {Ingrid Daubechies},
  title     = {Ten Lectures on Wavelets},
  series    = {CBMS-NSF Regional Conference Series in Applied Mathematics},
  volume    = {61},
  publisher = {Society for Industrial and Applied Mathematics},
  address   = {Philadelphia, PA},
  year      = {1992},
  doi       = {10.1137/1.9781611970104}
}
\appendix

\section{Comparison with the Alternative Definition of \texorpdfstring{\(\sigma_{\mathsf{SO}}\)}{$\sigma_{\mathsf{SO}}$} in \cite{BJ25a}}\label{subsec:comparison-three-so-forms}

Besides the Fourier definition of the smoothed-out variation in
\Cref{sec:prelims}, 
\cite[Section~5.3]{BJ25a} gave an alternative definition based
on a randomly shifted dyadic decomposition. This formulation is
closely related to, but different from, the shift-averaged Haar--Besov
seminorm in
\Cref{thm:shift-averaged-mixed-haar-equivalence}.

To make the comparison explicit, consider \(d=1\) and a smooth
periodic function \(f\). Let $f_{s}(z) = f((z + s) \bmod 1)$. Denote the quantity in
\cite[Section~5.3]{BJ25a} by
\(\mathcal V_{\mathsf{BJ}}(f)\). By
\cite[Claim~5.1]{BJ25a},
\begin{equation}
\label{eq:BJ_norm_1d}
    \mathcal V_{\mathsf{BJ}}(f)^2
    :=
    \lim_{h\to\infty}
    \int_0^1
    \sum_{j=1}^h
    \sum_{\ell=0}^{2^j-1}
    \left|
        f_s\bigl((\ell+1)2^{-j}\bigr)
        -
        f_s\bigl(\ell2^{-j}\bigr)
    \right|^2\,ds.
\end{equation}
Thus, \(\mathcal V_{\mathsf{BJ}}\) measures endpoint increments over
randomly shifted dyadic intervals.

By contrast, the
Haar--Besov seminorm measures differences between averages over adjacent
dyadic intervals. Recall that \(B_{j,\ell}\) denotes the indicator of the $\ell^{\text{th}}$ dyadic interval at level $j$. For any function \(g\), write $g_{j,\ell}
    :=
    2^j\langle g,B_{j,\ell}\rangle$
for its average over that interval. Since
\[
    2^j
    \left|
        \langle f_s,\psi_{j,\ell}\rangle
    \right|^2
    =
    \frac14
    \left|
       (f_s)_{j+1,2\ell}
-
(f_s)_{j+1,2\ell+1}
    \right|^2,
\]
we have
\begin{equation}
\label{eq:Haar_norm_1d}
 \mathbb E_{s\sim[0,1)}
\bigl[\lVert f_s\rVert_{\mathsf H}^2\bigr]
    =
    \frac14
    \int_0^1
    \sum_{j\ge0}
    \sum_{\ell=0}^{2^j-1}
    \left|
     (f_s)_{j+1,2\ell}
-
(f_s)_{j+1,2\ell+1}
    \right|^2\,ds.
\end{equation}

Although the two dyadic quadratic forms are distinct, shift averaging
makes both diagonal in the Fourier basis. Moreover, their Fourier
multipliers are comparable to \(|k|\). Indeed, Parseval's identity gives
(see \cite[Lemma~5.3]{BJ25a})
\[
    \mathcal V_{\mathsf{BJ}}(f)^2
    =
    \sum_{k\ne0}
    W_{\mathsf{BJ}}(k)
    |\widehat{f}_k|^2,
    \qquad
    W_{\mathsf{BJ}}(k)
    :=
    4\sum_{j\ge1}
    2^j
    \sin^2\left(\frac{\pi k}{2^j}\right),
\]
whereas the proof of
\Cref{thm:shift-averaged-mixed-haar-equivalence}  yields
\[
   \mathbb E_{s\sim[0,1)}
\bigl[\lVert f_s\rVert_{\mathsf H}^2\bigr]
    =
    \sum_{k\ne0}
    W_{\mathsf H}(k)
    |\widehat{f}_k|^2,
    \qquad
    W_{\mathsf H}(k)
    :=
    \frac4{\pi^2k^2}
    \sum_{j\ge0}
    2^{3j}
    \sin^4\left(\frac{\pi k}{2^{j+1}}\right).
\]
A direct calculation shows that
\[
    W_{\mathsf{BJ}}(k)
    \asymp
    W_{\mathsf H}(k)
    \asymp
    |k|,
    \qquad k\ne0.
\]
Consequently,
\[
    \mathcal V_{\mathsf{BJ}}(f)^2
    \asymp
    \mathbb E_{s\sim[0,1)}
\bigl[\lVert f_s\rVert_{\mathsf H}^2\bigr]
    \asymp
    \sigma_{\mathsf{SO}}(f)^2.
\]
Thus, the three formulations characterize the same function space, although they
encode it through different quadratic forms.

The advantage of the Haar formulation for our analysis lies not in
defining a different smoothness class, but in its direct compatibility
with the dyadic discrepancies controlled by the algorithm of \cite{BJ25a}. Indeed,
\[
  \mathsf{disc}_t(\psi_{j,\ell}) = \sum_{z\in A_t}x_t(z)\psi_{j,\ell}(z)
    =
    2^{j/2}\left(
       \mathsf{disc}_t(B_{j+1,2\ell})
-
\mathsf{disc}_t(B_{j+1,2\ell+1})
    \right).
\]
Thus, after expanding \(P_{<h}f_s\) in the Haar basis,
\(\mathsf{disc}_t(P_{<h}f_s)\) becomes a linear functional of the
dyadic discrepancy vector
\[
    \bigl(\mathsf{disc}_t(B)\bigr)_{B\in\mathcal D(h)},
\]
with coefficients whose squared Euclidean norm is controlled by
\(\lVert P_{<h}f_s\rVert_{\mathsf H}^2\).
The sub-Gaussianity in \Cref{prop:sub-Gaussianity} therefore controls
this linear functional, as made precise in
\Cref{lem:low-frequencyDiscrepancyBound}. Together with
\Cref{prop:TransferenceErrorBound}, this yields the desired
integration-error bound.

By comparison, as reviewed in \Cref{subsec:ReviewBJ25}, the analysis
in \cite{BJ25a} begins with the Hlawka--Zaremba formula, which expresses
the integration error in terms of the continuous discrepancy of prefix
intervals. After discretization, it writes $ \mathbf D_{\mathsf{prefix}}
    =
    P\mathbf D_{\mathsf{dyadic}}$, 
where \(P\) is the dyadic decomposition matrix. Consequently,
\[
    \left\langle
        \mathbf f',
        \mathbf D_{\mathsf{prefix}}
    \right\rangle
    =
    \left\langle
        P^\top\mathbf f',
        \mathbf D_{\mathsf{dyadic}}
    \right\rangle .
\]
The sub-Gaussian control of
\(\mathbf D_{\mathsf{dyadic}}\) therefore reduces their analysis to
bounding \(\lVert P^\top\mathbf f'\rVert_2\). Their structural analysis
of \(P\) bounds this quantity through the endpoint-increment quadratic
form appearing in \eqref{eq:BJ_norm_1d}. In our approach, by contrast,
the Haar expansion expresses
\(\mathsf{disc}_t(P_{<h}f_s)\) directly as a linear functional of the
dyadic discrepancy vector, without passing through either the
continuous discrepancy function or the prefix-discrepancy vector.

\end{document}